%% file: main.tex
\documentclass[11pt,a4paper]{article}

\usepackage[T1]{fontenc}
\usepackage[utf8]{inputenc}
\usepackage{lmodern}
\usepackage{microtype,fullpage}

\usepackage{amsmath,amssymb,amsthm,mathtools}
\usepackage{thm-restate}
\DeclarePairedDelimiter\ceil{\lceil}{\rceil}

\usepackage{graphicx}
\usepackage{xcolor}
\usepackage{tikz}
\usetikzlibrary{
  arrows.meta,
  positioning,
  shapes.geometric,
  shapes.misc,
  fit,
  calc,
  math,
  shapes.multipart,
  decorations.text,
  decorations.markings,
  decorations.pathmorphing,
  decorations.pathreplacing,
}
\usepackage{tikz-dimline}

\usepackage[vlined]{algorithm2e}

\usepackage[hidelinks]{hyperref}
\usepackage[nameinlink]{cleveref}

\usepackage{enumitem}

\usepackage{comment}
\usepackage{epstopdf}
\usepackage{etoolbox}
\usepackage{xifthen}
\usepackage{boxedminipage}
\usepackage{framed}
\usepackage{tabularx}
\usepackage{tcolorbox}

\usepackage{
graphicx,
mathtools
}

\usepackage{ 
boxedminipage,
etoolbox,
framed,
ifthen,
tabularx,
tcolorbox,
tikz,
xifthen
}

\newlength{\RoundedBoxWidth}
\newsavebox{\GrayRoundedBox}
\newenvironment{GrayBox}[1]%
   {\setlength{\RoundedBoxWidth}{.95\textwidth}
    \def\boxheading{#1}
    \begin{lrbox}{\GrayRoundedBox}
       \begin{minipage}{\RoundedBoxWidth}}%
   {   \end{minipage}
    \end{lrbox}
    \begin{center}
    \begin{tikzpicture}%
       \node(Text)[draw=black!20,fill=white,rounded corners,%
             inner sep=2ex,text width=\RoundedBoxWidth]%
             {\usebox{\GrayRoundedBox}};
        \coordinate(x) at (current bounding box.north west);
        \node [draw=white,rectangle,inner sep=3pt,anchor=north west,fill=white] 
        at ($(x)+(6pt,.75em)$) {\boxheading};
    \end{tikzpicture}
    \end{center}}
\newenvironment{defproblemx}[2][]{\noindent\ignorespaces%
                \FrameSep=6pt%
                \parindent=0pt%
                \vspace*{-1em}
                \ifthenelse{\isempty{#1}}{%
                  \begin{GrayBox}{\textsc{#2}}%
                }{%
                  \begin{GrayBox}{\textsc{#2}  parameterized by~{#1}}%
                }
                \begin{tabular*}{\textwidth}{@{\hspace{.1em}} >{\itshape} p{1cm} p{0.8\textwidth} @{}}%
            }{
                \end{tabular*}%
                \end{GrayBox}%
                \ignorespacesafterend
            }  
\newcommand{\defproblem}[3]{
  \begin{defproblemx}{#1}
    Input:  & #2 \\
    Task: & #3
  \end{defproblemx}
}%
\newtheorem{theorem}{Theorem}[section]
\newtheorem{lemma}[theorem]{Lemma}

\newtheorem{corollary}[theorem]{Corollary}
\newtheorem{observation}[theorem]{Observation}
\newtheorem{claim}{Claim}

\theoremstyle{definition}
\newtheorem{definition}[theorem]{Definition}

\theoremstyle{remark}
\newtheorem{remark}[theorem]{Remark}

\crefname{algorithm}{Algorithm}{Algorithms}
\Crefname{algorithm}{Algorithm}{Algorithms}

\crefname{figure}{Figure}{Figures}
\Crefname{figure}{Figure}{Figures}

\crefname{theorem}{Theorem}{Theorems}
\Crefname{theorem}{Theorem}{Theorems}

\crefname{lemma}{Lemma}{Lemmas}
\Crefname{lemma}{Lemma}{Lemmas}

\crefname{section}{Section}{Sections}
\Crefname{section}{Section}{Sections}

\crefname{claim}{Claim}{Claims}
\Crefname{claim}{Claim}{Claims}

\crefname{observation}{Observation}{Observations}
\Crefname{observation}{Observation}{Observations}

\crefname{proposition}{Proposition}{Propositions}
\Crefname{proposition}{Proposition}{Propositions}

\newenvironment{claimproof}{\proof[Proof of claim]}{\endproof}

\newcommand{\pendant}[1]{\left\lceil\frac{#1}{2}\right\rceil}

\newcommand{\tc}{\textsc{TriConn}}

\title{Connectivity Augmentation of Plane Graphs}

\author{
Krishnan Dehaleesan \\
University of Bergen, Norway\\
\texttt{krishnan.dehaleesan@uib.no}
\and
Asif Khan\thanks{Partially funded by a grant from Infosys foundation.}\\
Chennai Mathematical Institute, India\\
\texttt{asifkhan@cmi.ac.in}
\and
Pranabendu Misra\\
Chennai Mathematical Institute, India\\
\texttt{pranabendu@cmi.ac.in}
}

\date{}
\begin{document}

\maketitle

\begin{abstract}
    We study the problem of connectivity augmentation of a planar graph, while preserving planarity. This problem is motivated by many real-world settings such as road-networks, power-networks etc. In these settings, it is crucial to preserve the original planar embedding after augmentation.

	In 2009, Gutwenger and Mutzel gave a constructive algorithm showing that a connected planar graph with a fixed embedding (a plane graph) can be optimally augmented to a biconnected graph without crossings while preserving the embedding. We further this line of research, by giving an algorithm that computes a minimum set of edges that makes a connected plane graph 2-edge-connected in \(O(|V|(1+\alpha(|V|)))\) time and linear space, where \(\alpha\) is the inverse Ackermann function.

    We also study the 3-vertex-connectivity augmentation of biconnected outerplanar plane graphs. We present the first polynomial-time algorithm that augments such graphs to 3-connectivity with the minimum number of edges in \(O(|V|(1+\alpha(|V|)))\) time and linear space while preserving the embedding, i.e., the augmented graph has a planar embedding that extends the given embedding.
 
\end{abstract}

\input{Sections/1.Introduction}

\input{Sections/2.1.prelims}

\input{Sections/2.bridgeconnectivityaug}
\input{Sections/3.Triconnforfixedembed}

\input{Sections/4.almostopt}

\input{Sections/5.opt}

\input{Sections/6.conclusion}

\bibliographystyle{plainurl}
\bibliography{cas-refs}

\end{document}

%% file: Sections/1.Introduction.tex
 \section{Introduction}

The \textbf{connectivity augmentation} problem asks to increase the connectivity of a given graph to a specified level while adding as few edges as possible. This is a classic problem in graph algorithms and network design, motivated by applications in improving connectivity of real world networks, such as roads, electricity, telecommunication, and so on. While it is certainly possible to augment the connectivity by increasing the capacity of existing links, or adding parallel links, this is sub-optimal in both the costs, and in improving the network resiliency to link failures. For example, if a link in an electricity-network has failed due to some real-world incident (e.g a falling tree), it is highly likely that any parallel links would have also have failed. Hence, it is highly desirable to have connectivity via some alternative route. 

This motivates the algorithmic study of connectivity augmentation problems, which play a fundamental role in the design of reliable communication and infrastructure networks.
For general graphs, both biconnectivity augmentation (making a graph 2-vertex-connected) and bridge-connectivity augmentation (making a graph 2-edge-connected), with uniform costs for new-links, admit linear-time algorithms \cite{eswaran1976augmentation}. However, they become NP-hard when new links have different costs.

A variant of connectivity augmentation problems considers the case where the input graph is planar and we want the augmented graph to be planar as well. This is directly motivated by real-world settings such as augmentation of road networks, electricity networks, etc. that are always embedded on a plane, and hence a plane solution is highly preferred. 
This variant, called \textsc{Planar BiConnectivity Augmentation (PVCA)}, is known to be NP-hard~\cite{kant1991planar}, even with uniform new-link costs. The edge-connectivity version, \textsc{Planar BridgeConnectivity Augmentation (PECA)} is also NP-hard~\cite{wolff2012augmenting}.

Let us observe that, in many motivating real-world problems, such as augmenting a road-network, not only is the input graph planar, but it has a fixed embedding of the nodes in the plane which cannot be changed. Therefore, it is highly desirable that upon augmenting the graph with new-links, the embedding is preserved.
 
This motivates us to study the following problem:
\defproblem
{Planar Bridgeconnectivity Augmentation for Fixed Embedding (PECA-Fix)}
{A connected planar graph $G$ with a fixed embedding.}
{Find a minimum set of edges $E' \subset V \times V$ such that $G+E'$ has a planar embedding that extends the planar embedding of $G$ and is 2-edge connected.}

The vertex-connectivity variant of PECA-Fix, called PVCA-Fix, can be defined similarly. This was studied by Gutwenger, Mutzel and Zey \cite{gutwenger2009planar} who showed that solving biconnectivity augmentation for planar graphs with a fixed embedding (PVCA-Fix) is polynomial-time solvable in near-linear time.  Since plane biconnectivity augmentation was poly-time solvable, it was natural to ask the question for connectivity greater than 2. Recent work by Akitaya et al.~\cite{akitaya2025price} showed that, even for plane graphs, increasing connectivity from $c$ to $k$ is NP-hard for all $2 \le c < k \le 5$. This motivated us to study the tri-connectivity augmentation of graphs from a subclass of planar graphs, namely outerplanar graphs, with a fixed embedding:

\defproblem
{Outerplanar Triconnectivity Augmentation for Fixed Embedding (OTA-Fix)}
{An Outerplane biconnected graph $G$ with a fixed embedding.}
{Find a minimum set of edges $E' \subset V \times V$ such that $G+E'$ has a planar embedding that extends the planar embedding of $G$ and is 3-vertex-connected.}

Note that the requirement is only that the augmented graph admits a planar embedding extending the fixed embedding of $G$; the augmented graph need not remain outerplanar. Indeed, except for a triangle, no 3-connected graph is outerplanar.

Let us note that the triconnectivity augmentation problem for outerplanar graphs (without a fixed embedding) is well-studied. The authors of  \cite{garcia2010augmenting} gave a polynomial time algorithm for the biconnectivity augmentation problem and the bridge-connectivity augmentation problem.  For $k$ even or $k=3$, Nagamochi and Eades \cite{nagamochi1998edge} provided optimal outerplanar edge-connectivity augmentation algorithms.
Furthermore, Kant \cite{kant1996augmenting} provided an optimal solution in polynomial time for the Outerplanar triconnectivity augmentation problem. 

The fixed embedding constraint rules out the known techniques for OTA-Fix. For instance, Kant's exact algorithm that augments an outerplanar graph to triconnectivity does not extend to the setting where the input graph has a fixed embedding. In particular, Kant’s method may alter the embedding to a non-outerplanar one while preserving planarity, and then insert edges accordingly. To illustrate this limitation, consider the graph obtained from a cycle $C_k$ by replacing each edge with a parallel edge and subdividing it. Kant’s algorithm can make this outerplanar graph 3-connected using $\frac{k}{2}$ edges by effectively moving vertices from the outer face into inner faces. Such modifications are not allowed in the OTA-Fix setting, where the embedding must be preserved. In fact, for the example above, no augmentation respecting the fixed embedding can achieve 3-connectivity using only $\frac{k}{2}$ edges.

This raises an interesting question: ``how much worse can the solution be when the embedding is fixed?'' 
For instance, in the case of 2-connectivity augmentation of plane graphs, there exist graphs where a fixed embedding requires $k$ edges, whereas allowing changes to the embedding admits a solution of size $\frac{k}{2}$.
Our results show that this gap is significantly smaller for outerplanar graphs: although the optimal augmentation without embedding constraints uses $\frac{k}{2}$ edges, in the fixed-embedding setting the optimum increases by at most one, i.e., it is bounded by $\frac{k}{2} + 1$.

\subsection{Our Contributions}

\textbf{2-edge Connectivity Augmentation for Plane Graphs (PECA-Fix)}: We present the first polynomial-time algorithm to optimally augment a connected planar graph (with a fixed embedding) to eliminate all bridges. Our algorithm runs in $O(n \cdot (1+\alpha(n)))$ time using $O(n)$ space, where $n=|V|$ and $\alpha(\cdot)$ is the inverse Ackermann function (arising from union–find operations). We build upon the framework of~\cite{gutwenger2009planar}, meant for vertex-connectivity, to handle our edge-connectivity setting.

\medskip
\noindent
\textbf{Triconnectivity Augmentation for Outerplane Graphs (OTA-Fix)}: Our main result is an $O(n \cdot (1+\alpha(n)))$ time exact algorithm for augmenting a biconnected outerplanar graph (with fixed embedding) to achieve 3-vertex-connectivity. 

In a setting where problems are NP-hard for connectivity greater than 2, we prove that Outerplane graphs form a sufficiently rich restriction admitting a fast polynomial time algorithm. In addition, We show that despite the embedding constraint, the optimal solution size increases by at most one in comparison with the analogous problem without a fixed embedding.

\medskip
\noindent

A central idea in our algorithm for OTA-fix is a face-splitting technique that could possibly be useful in other contexts as well. We identify a “large” face in the embedding that hinders connectivity, and we add internal edges to split this face into smaller ones, which simplifies the augmentation task. Moreover, for the outerplanar augmentation problem, we prove a noteworthy structural property of optimal solutions: at most two new edges need to be added inside the outerplanar graph in any optimum augmentation, no matter how large the graph is. This property significantly restricts the search space of possible solutions and underpins the correctness and optimality of our algorithm.

\subsection{Related Work}

Biconnectivity augmentation is the special case of $k$-connectivity augmentation for $k=2$. Polynomial-time algorithms are known for $k=3$~\cite{hsu1991linear,watanabe1993minimum} and $k=4$~\cite{hsu1992four}, while the problem remained widely open for $k\ge 5$. This was resolved when, for a fixed $k$, Jackson et al.~\cite{jackson2005independence} gave a polynomial-time algorithm, and V\'{e}gh~\cite{DBLP:journals/siamdm/Vegh11} showed that a $(k-1)$-connected graph can be augmented to $k$-connectivity in polynomial time. The analogous edge version includes bridge-connectivity augmentation ($k=2$), and general $k$-edge-connectivity augmentation for which Watanabe~\cite{watanabe1987edge} provided a polynomial-time algorithm. 

For planar biconnectivity augmentation (PVCA), Kant and Bodlaender~\cite{kant1991planar} gave a 2-approximation and claimed a $3/2$-approximation; counterexamples were later shown by Fialko and Mutzel~\cite{fialko1998new}, who proposed a $5/3$-approximation, which was also refuted by Gutwenger, Mutzel and Zey~\cite{gutwengerhardness}. Thus, the best known ratio remains 2. Geometric variants of biconnectivity augmentation have also been studied~\cite{altri,ABELLANAS2008220,toth2012connectivity}. Other works on planar graphs include augmentation problems where the output is required to be regular (or cubic)~\cite{hartmann2012cubic, hartmann2015regular}.

\paragraph*{Organization of the paper.} In \cref{bridge}, we give a polynomial-time algorithm for PECA-Fix in connected planar graphs with a fixed embedding. In \cref{triconnproblem}, we study our main contribution to this paper: the OTA-Fix algorithm and its analysis. We first provide an algorithm that uses slightly more than the optimal number of augmentation edges in \cref{almostopt}. 
Then, in \cref{optalg} we use this to provide an exact polynomial time algorithm.

%% file: Sections/2.1.prelims.tex
\section{Preliminaries}

For a graph $G$, $V(G)$ and $E(G)$ denote its vertex and edge sets. We use standard graph-theoretic terminology as in~\cite{diestel2017} for connectivity, cut-edges (also called bridges), cut-vertices and planar graphs.
In particular, a graph is $k$-vertex-connected if it has more than three vertices and remains connected after deleting any set of at most $k-1$ vertices. Similarly, a graph is $k$-edge-connected if it is connected and remains so even when any subset of edges of up to $k-1$ size is deleted.  A \emph{plane graph} is a planar graph equipped with a fixed
embedding in the plane.

 A subset $S\subseteq V(G)$ is a \emph{block} (or \emph{biconnected component}) if $G[S]$ is biconnected and maximal with this property. A \emph{bridge-component} is a maximal vertex set $S$ such that $G[S]$ is $2$-edge connected.

The \emph{block-cut tree} ($bc$-tree) of a connected graph $G$, denoted $bc(G)$, is the graph whose nodes are the blocks (\emph{b-nodes}) and cut vertices (\emph{c-nodes}) of $G$, with a block node adjacent to a cut-vertex node iff the cut vertex lies in the block.

Similarly, the \emph{bridgetree} $bt(G)$ has nodes corresponding to bridge-components (b-nodes) and cut-edges (c-nodes), where a bridge-component node is adjacent to a cut-edge node iff an endpoint of the cut-edge lies in the component. Every c-node in $bt(G)$ has degree $2$.

The blocks corresponding to the leaf b-nodes in $bc(G)$ are called \emph{pendant blocks}. Pendant bridge-components are defined analogously in $bt(G)$.

\begin{definition}[Massive node \& Balanced BC-tree]\label{defnbalancedbc}
    Given a connected graph $G$,
    let $p$ be the number of leaves (pendants) in $bc(G)$. A c-node $c^*$ in $bc(G)$,  is called \emph{massive} if $\deg_{bc(G)}(c^*) \ge \pendant{p}+2$. The $bc$-tree of $G$ is called \emph{balanced} if it does not contain any massive c-node.
    Otherwise, $bc(G)$ is called \emph{unbalanced}.
\end{definition}

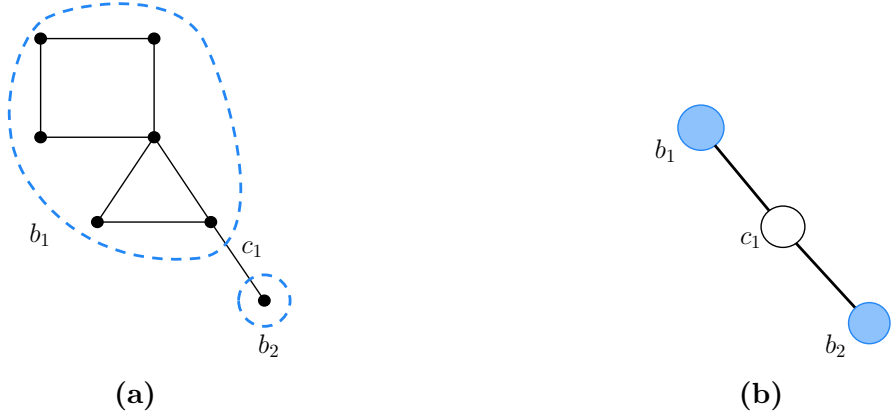
\begin{figure}
  \centering
  \begin{minipage}[t]{0.48\textwidth}
    \centering
    \input{Figs/tikzfigs/edgegraph.tex}

    \smallskip
    \textbf{(a)}
    \label{fig:sfig1}
  \end{minipage}\hfill
  \begin{minipage}[t]{0.48\textwidth}
    \centering
    \input{Figs/tikzfigs/bridgetree.tex}

    \smallskip
    \textbf{(b)}
    \label{fig:sfig2}
  \end{minipage}

  \caption{An example of representing a $bt$-tree. (a) Original graph with a fixed embedding; (b) The bridge tree of the graph.}
  \label{fig:bridgetree}
\end{figure}

A bridge-component of $G$ is called \emph{pendant} if the node corresponding to it in the bridge-tree of $G$ is a leaf. 

Recall that $G$ has a fixed planar embedding. Two bridge-components $B_1,B_2$ are
\emph{matchable} if they contain vertices on a common face, so that adding an edge between them preserves the embedding. The corresponding b-nodes in $bt(G)$ are also called matchable.

\begin{observation}\label{bridgetreechange}
Let $u,v$ be matchable pendant b-nodes in $bt(G)$ and let $G'=G+e$ where $e$ joins vertices in components of $u$ and $v$. Then $bt(G')$ is obtained from $bt(G)$ by contracting the $u$–$v$ path in $bt(G)$ into a single b-node $b'$, preserving all adjacencies outside the path.
\end{observation}

\begin{observation}[Lower bound]\label{lowerbound}
If $p$ is the number of pendants in $bt(G)$, at least $\lceil p/2\rceil$ edges are required to make $G$ 2-edge-connected.
\end{observation}

For any two pendants $p_1$ and $p_2$, adding the edge between them reduces the number of pendants remaining at least by 1. We shall now define a condition for bridgetrees analogous to the leaf connectivity condition in \cite{gutwenger2009planar}.

\begin{definition}[Leaf connection condition (l.c.c)]\label{llc}
Two pendants $p_1,p_2$ in $bt(G)$ satisfy the l.c.c. if the path between them contains either two b-nodes of degree at least three or one b-node of degree at least four.
\end{definition}

\begin{lemma}\label{decreaseby2}
If $p_1,p_2$ satisfy the l.c.c., adding an edge between them reduces the number of pendants by two.
\end{lemma}

\medskip
\emph{Bundles} and \emph{labels} were defined with the context of $bc$-tree in  \cite{fialko1998new}. We shall do it for our requirement in the $bt$-tree now. 

\medskip
\noindent\textbf{Bundles and labels.}
Let $P$ be the set of pendants of $bt(G)$. A set $B\subseteq P$ is a
\emph{bundle} if (i) every pair in $B$ is matchable, (ii) joining any pair produces a new pendant in $bt(G+e)$, and (iii) $B$ is maximal. By \cref{bridgetreechange}, two pendants lie in the same bundle iff the path between them in $bt(G)$ contains exactly one degree-3 b-node, called the \emph{parent} of the bundle. A bundle together with its parent is a \emph{label}; its size $L(\ell)$ is the number of pendants the bundle contains.

\medskip
\noindent\textbf{Intuition.} 
Inside a bundle, any two pendants are mutually matchable and behave the same: connecting any pair does not eliminate the entire group, but produces a new pendant through their common parent. Thus, the bundle acts as a single unit attached to the rest of the bridgetree at its parent. 
Labels therefore partition the pendants into blocks that cannot be fully resolved internally; useful augmentation edges typically connect pendants from \emph{different} labels, which ensures progress (via the l.c.c.) in reducing the number of pendants.

\begin{restatable}{lemma}{lccsatisfied}\label{lccsatisfied}
If two matchable pendants belong to different labels, then they satisfy the l.c.c.
\end{restatable}

\begin{proof}
    We prove that if $p_1$ and $p_2$ don't satisfy the l.c.c, then they must belong to the same label; this implies the statement of the lemma.
    If the l.c.c is not satisfied, then the path between $p_1$ and $p_2$ will have only one degree $3$ b-vertex, $v$. In this case, adding an edge between $p_1$ and $p_2$ leaves the new coalesced block obtained from $v$ with exactly one adjacent c-node and hence $v$ is a new pendant, leaving us with only one less pendant. Hence, $p_1$ and $p_2$ belong to the same label.
\end{proof}

%% file: Figs/tikzfigs/edgegraph.tex
\scalebox{0.7}{\tikzset{every picture/.style={line width=0.75pt}} 
	
	\begin{tikzpicture}[x=0.75pt,y=0.75pt,yscale=-1,xscale=1]
		
		\draw   (137.57,45.37) -- (218.9,45.37) -- (218.9,116.08) -- (137.57,116.08) -- cycle ;
		\draw   (218.9,116.08) -- (259.57,176.86) -- (178.24,176.86) -- cycle ;
		\draw    (259.57,176.86) -- (297.92,233.15) ;
		\draw  [fill={rgb, 255:red, 0; green, 0; blue, 0 }  ,fill opacity=1 ] (133.48,45.37) .. controls (133.48,43.22) and (135.31,41.48) .. (137.57,41.48) .. controls (139.83,41.48) and (141.65,43.22) .. (141.65,45.37) .. controls (141.65,47.51) and (139.83,49.25) .. (137.57,49.25) .. controls (135.31,49.25) and (133.48,47.51) .. (133.48,45.37) -- cycle ;
		\draw  [fill={rgb, 255:red, 0; green, 0; blue, 0 }  ,fill opacity=1 ] (214.82,45.37) .. controls (214.82,43.22) and (216.65,41.48) .. (218.9,41.48) .. controls (221.16,41.48) and (222.99,43.22) .. (222.99,45.37) .. controls (222.99,47.51) and (221.16,49.25) .. (218.9,49.25) .. controls (216.65,49.25) and (214.82,47.51) .. (214.82,45.37) -- cycle ;
		\draw  [fill={rgb, 255:red, 0; green, 0; blue, 0 }  ,fill opacity=1 ] (133.48,116.08) .. controls (133.48,113.94) and (135.31,112.2) .. (137.57,112.2) .. controls (139.83,112.2) and (141.65,113.94) .. (141.65,116.08) .. controls (141.65,118.23) and (139.83,119.97) .. (137.57,119.97) .. controls (135.31,119.97) and (133.48,118.23) .. (133.48,116.08) -- cycle ;
		\draw  [fill={rgb, 255:red, 0; green, 0; blue, 0 }  ,fill opacity=1 ] (214.82,116.08) .. controls (214.82,113.94) and (216.65,112.2) .. (218.9,112.2) .. controls (221.16,112.2) and (222.99,113.94) .. (222.99,116.08) .. controls (222.99,118.23) and (221.16,119.97) .. (218.9,119.97) .. controls (216.65,119.97) and (214.82,118.23) .. (214.82,116.08) -- cycle ;
		\draw  [fill={rgb, 255:red, 0; green, 0; blue, 0 }  ,fill opacity=1 ] (174.15,176.86) .. controls (174.15,174.71) and (175.98,172.98) .. (178.24,172.98) .. controls (180.49,172.98) and (182.32,174.71) .. (182.32,176.86) .. controls (182.32,179) and (180.49,180.74) .. (178.24,180.74) .. controls (175.98,180.74) and (174.15,179) .. (174.15,176.86) -- cycle ;
		\draw  [fill={rgb, 255:red, 0; green, 0; blue, 0 }  ,fill opacity=1 ] (255.49,176.86) .. controls (255.49,174.71) and (257.32,172.98) .. (259.57,172.98) .. controls (261.83,172.98) and (263.66,174.71) .. (263.66,176.86) .. controls (263.66,179) and (261.83,180.74) .. (259.57,180.74) .. controls (257.32,180.74) and (255.49,179) .. (255.49,176.86) -- cycle ;
		\draw  [fill={rgb, 255:red, 0; green, 0; blue, 0 }  ,fill opacity=1 ] (293.83,233.15) .. controls (293.83,231) and (295.66,229.27) .. (297.92,229.27) .. controls (300.17,229.27) and (302,231) .. (302,233.15) .. controls (302,235.29) and (300.17,237.03) .. (297.92,237.03) .. controls (295.66,237.03) and (293.83,235.29) .. (293.83,233.15) -- cycle ;
		\draw  [color={rgb, 255:red, 35; green, 134; blue, 245 }  ,draw opacity=1 ][dash pattern={on 5.63pt off 4.5pt}][line width=1.5]  (141,31.03) .. controls (161,20.03) and (235,8.03) .. (255,48.03) .. controls (275,88.03) and (304.5,197.84) .. (249.5,202.84) .. controls (194.5,207.84) and (152,185.03) .. (127,145.03) .. controls (102,105.03) and (121,42.03) .. (141,31.03) -- cycle ;
		\draw  [color={rgb, 255:red, 35; green, 134; blue, 245 }  ,draw opacity=1 ][dash pattern={on 5.63pt off 4.5pt}][line width=1.5]  (279.47,233.15) .. controls (279.47,222.96) and (287.73,214.71) .. (297.92,214.71) .. controls (308.1,214.71) and (316.36,222.96) .. (316.36,233.15) .. controls (316.36,243.33) and (308.1,251.59) .. (297.92,251.59) .. controls (287.73,251.59) and (279.47,243.33) .. (279.47,233.15) -- cycle ;
		
		\draw (128,176) node [font = \Large] [anchor=north west][inner sep=0.75pt]   [align=left] {$\displaystyle b_{1}$};
		\draw (292,256) node [font = \Large][anchor=north west][inner sep=0.75pt]   [align=left] {$\displaystyle b_{2}$};
		\draw (280,189) node [font = \Large] [anchor=north west][inner sep=0.75pt]   [align=left] {$\displaystyle c_{1}$};

\end{tikzpicture}}

%% file: Figs/tikzfigs/bridgetree.tex
\scalebox{0.7}{\tikzset{every picture/.style={line width=0.75pt}} 
	
	\begin{tikzpicture}[x=0.75pt,y=0.75pt,yscale=-1,xscale=1]
		
		\draw   (209.6,155.22) .. controls (209.6,147.04) and (216.63,140.41) .. (225.3,140.41) .. controls (233.97,140.41) and (241,147.04) .. (241,155.22) .. controls (241,163.4) and (233.97,170.03) .. (225.3,170.03) .. controls (216.63,170.03) and (209.6,163.4) .. (209.6,155.22) -- cycle ;
		\draw  [color={rgb, 255:red, 35; green, 134; blue, 245 }  ,draw opacity=1 ][fill={rgb, 255:red, 142; green, 194; blue, 252 }  ,fill opacity=1 ] (150,84.23) .. controls (150,75.29) and (157.39,68.03) .. (166.5,68.03) .. controls (175.61,68.03) and (183,75.29) .. (183,84.23) .. controls (183,93.18) and (175.61,100.44) .. (166.5,100.44) .. controls (157.39,100.44) and (150,93.18) .. (150,84.23) -- cycle ;
		\draw [line width=1.5]    (176.5,96.84) -- (215.5,143.84) ;
		\draw [line width=1.5]    (234.5,166.84) -- (287.15,224.31) ;
		\draw  [color={rgb, 255:red, 35; green, 134; blue, 245 }  ,draw opacity=1 ][fill={rgb, 255:red, 142; green, 194; blue, 252 }  ,fill opacity=1 ] (272.3,224.31) .. controls (272.3,216.18) and (278.95,209.6) .. (287.15,209.6) .. controls (295.35,209.6) and (302,216.18) .. (302,224.31) .. controls (302,232.44) and (295.35,239.03) .. (287.15,239.03) .. controls (278.95,239.03) and (272.3,232.44) .. (272.3,224.31) -- cycle ;
		
		\draw (132,91) node [font = \Large] [anchor=north west][inner sep=0.75pt]   [align=left] {$\displaystyle b_{1}$};
		\draw (193,158) node [font = \Large] [anchor=north west][inner sep=0.75pt]   [align=left] {$\displaystyle c_{1}$};
		\draw (254,231) node [font = \Large][anchor=north west][inner sep=0.75pt]   [align=left] {$\displaystyle b_{2}$};
		
\end{tikzpicture}}

%% file: Sections/2.bridgeconnectivityaug.tex
\section{Bridgeconnectivity augmentation for a fixed embedding}\label{bridge}
In this section, we give an algorithm that solves PECA-Fix. 
Recall that, in PECA-Fix, the input is a connected plane graph $G$ and the goal is to find a subset of
edges $E' \subseteq V(G) \times V(G) $, of minimum size such that $G' = G+E'$ is bridgeconnected and planar.

Let us restate our result for the sake of convenience.

\begin{restatable}{theorem}{twoedgeconnopt}
\label{thm_2edgeconnopt}
 There exists an  algorithm that for a given graph $G$ solves PECA-Fix in $\mathcal{O}(|V(G)|(1+\alpha(|V(G)|)))$ time and requires $\mathcal{O}(|V(G)|)$ space.
\end{restatable}

\subsection{Preliminaries and Structural Remarks}

For a plane graph $G$ which is a tree, we can provide a \emph{cyclic ordering} of the leaves in this tree from its embedding. We can start from some leaf node and move along the outer face (the only face for a tree) in clockwise direction and add the leaves whenever they are visited. This provides a cyclic ordering, $\pi$, of the leaves in this tree. Two pendants are \emph{adjacent} if they appear successively in the cyclic ordering. Two labels are said to be \emph{neighbours} if a pendant of one is adjacent to a pendant of another. Each label has an adjacent label.

For a face $F$, consider the subgraph of $G$ that consists of vertices and edges that lie on $F$, call it $\textsc{facegraph}(F)$ (for brevity, let $H:=\textsc{facegraph}(F)$). 
 
 Since augmentation edges lie inside faces, the problem reduces to making each facegraph bridgeconnected. We therefore operate on the bridgetree of each facegraph. The embedding of $bt(H)$ is inherited from the embedding of $H$ by preserving the cyclic order of incident cut edges at each block.

\subsection{Algorithm for PECA-Fix}

The algorithm works by computing the optimal edges to eliminate all the bridges in each face $F$. For each face $F$, \cref{Face_conn_edge} computes $H$, and iteratively connects two ‘pendant’ components to reduce their number (hinging on the approach of \cite{gutwenger2009planar} for biconnectivity). 

Specifically, the following is done in $H$: we pick a maximum sized label and connect one of its pendant to an adjacent pendant from a different label. The algorithm keeps repeating this process over and over again until there is only one label, and then we handle one of two possible cases: the remaining label has size 3, or the remaining label has size 2. In each of the cases we find the optimal solution. Refer \Cref{fig:facegraph} for an illustrative example of \textsc{Face\_Conn}.

\begin{algorithm}[t]
\caption{\textsc{PA\_Bridgeconn}}
\label{alg:Bridgeconn}

\KwIn{A connected plane graph $G$}
\KwOut{A set of vertex pairs $E'$ such that $G + E'$ is $2$-edge connected}

$E' \gets \emptyset$\;

\For{each face $F$}{
    $A \gets \textsc{Face\_Conn}(F)$\;
    $E' \gets E' \cup A$\;
}

\Return{$E'$}\;

\end{algorithm}

\begin{algorithm}[t]
\caption{\textsc{Face\_Conn}}
\label{Face_conn_edge}

\KwIn{A face $F$ of a graph with a fixed embedding}
\KwOut{List of edges $A$ that 2-edge-connect the face $F$}

$H \gets \textsc{facegraph}(F)$\;
Compute $bt(H)$\;
$A \gets \emptyset$\;

\While{number of labels $> 1$}{
    Consider $bt(F)$ and take any label $l_1$, and another neighbouring label $l_2$\;
    Let $p_1$ from $l_1$ and $p_2$ from $l_2$ represent pendant blocks adjacent to each other in the cyclic ordering of pendants\;
    Add an edge $e$ in $H$ connecting vertices of the blocks $p_1$ and $p_2$\;
    Update $H \gets H + e$\;
    $A \gets A \cup \{e\}$\;
}

\If{there is exactly one label $l$ in $H$}{
    \eIf{$\mathrm{size}(l) = 2$}{
        Let $e$ be an edge connecting a vertex from each pendant block\;
        $A \gets A \cup \{e\}$\;
        
    }{
        \If{$\mathrm{size}(l) = 3$}{
            Add $2$ edges $e_1$ and $e_2$ between pendants $p_1,p_2$ and $p_2,p_3$\;
            $A \gets A \cup \{e_1, e_2\}$\;
            
        }
    }
}

\Return{$A$}\;
\end{algorithm}

\begin{figure}
  \centering
  \begin{minipage}[t]{0.45\textwidth}
    \centering
    \includegraphics[scale = 0.7]{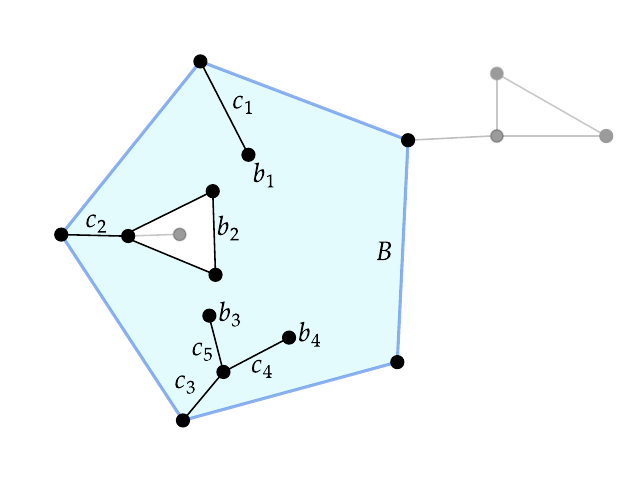}
    \smallskip

    \textbf{(a)}
    \label{fig1}
  \end{minipage}\hfill
  \begin{minipage}[t]{0.45\textwidth}
    \centering
    \includegraphics[scale = 0.6]{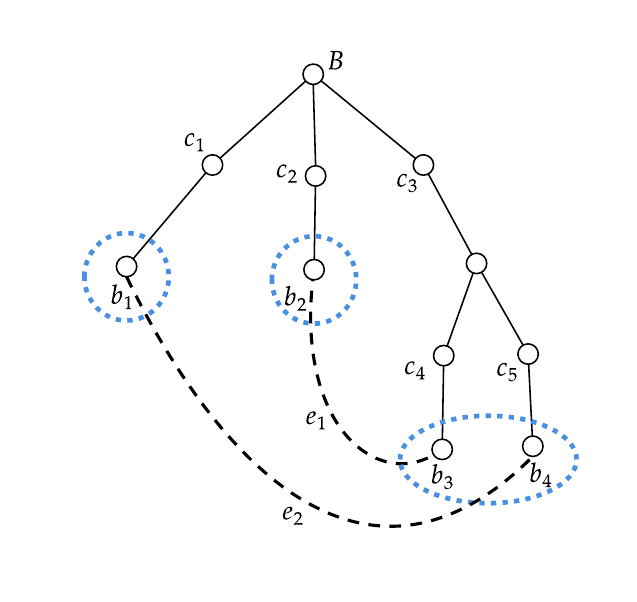}
    \smallskip

    \textbf{(b)}
    \label{fig2}
  \end{minipage}

  \caption{(a) A connected graph with a particular highlighted face containing cut edges (denoted by $c_i$'s) and bridge-components (denoted by $b_i$'s);
  (b) The bridge-tree decomposition of the \textsc{facegraph}, with dotted light-blue edges denoting the labels.
  The dotted black edges $e_1$ and $e_2$ show the augmentation edges between two neighbouring labels, in the first and second round of \textsc{Face\_Conn} respectively.}
  \label{fig:facegraph}
\end{figure}

\begin{restatable}{lemma}{bridgeconnplanar}\label{lem:bridgeconnplanar}
	{\sc PA\_Bridgeconn} outputs a set of edges $E'$ such that $G+E'$ is planar, extends the planar embedding of $G$, and is bridgeconnected.
\end{restatable}
\begin{proof}

Each edge joins vertices on a common face and is inserted between adjacent
pendant components in the cyclic order, so it can be drawn inside the face
without crossings. The process continues until no cut-edge remains and the statement follows.

\end{proof}

What is left is to prove that within each face, our algorithm finds the optimal number of edges to add. We use this to conclude that the added edges in total comprise the optimal solution.

\begin{restatable}{lemma}{faceopt}\label{lem_faceopt}

	For each face $F$ such that $bt(\textsc{Facegraph}(F))$ has $p$ pendants, \textsc{Face\_Conn} uses $\pendant{p}$ (optimal) edges to 2-edge connect \textsc{Facegraph}$(F)$.

\end{restatable}

\begin{proof}
	Consider a face $F$, and let $p$ denote the number of pendants.

If $p=2$, Algorithm \textsc{Face\_Conn} connects the two pendants by adding one edge, regardless of whether their parent in the $bt$-tree is a $b$-node or a $c$-node. Thus, this case is resolved.

If $p=3$, there must be a $b$-node of degree $3$ in the $bt$-tree (since every $c$-node has degree $2$). This is the only possible node of degree $3$. By the contrapositive of \cref{lccsatisfied}, all three pendants have the same label. The algorithm therefore adds two edges to augment the graph.

\textbf{Inductive hypothesis:} For any bridgetree with $p' < p$ pendants, the face can be made bridgeconnected using $\ceil{\frac{p'}{2}}$ edges.

Now assume $p>3$. Then there are at least two distinct labels (each $b$-label has size at most $3$). The algorithm adds an edge between pendants of two different labels. By \cref{lccsatisfied} and \cref{decreaseby2}, this reduces the number of pendants by $2$.

By the inductive hypothesis, $\ceil{\frac{p'}{2}} = \ceil{\frac{p-2}{2}}$ additional edges suffice to bridgeconnect the resulting graph $F+e$. Hence, the total number of edges used is
\[
\pendant{p-2}+1 = \pendant{p},
\]
which completes the proof.
\end{proof}

Combining the arguments in \Cref{lem_faceopt,lem:bridgeconnplanar} and arguing for the running time, we get the final theorem.

\twoedgeconnopt*
\begin{proof}
\textsc{PA\_BrideConn} provides the correct solution (\Cref{lem_faceopt}).
We compute the initial $bt$-tree of each face $F$ using a DFS-based
decomposition in $\mathcal{O}(|V_F|)$ time, since planar graphs have a
linear number of edges.

As edges are added during the augmentation, bridge-components merge.
We maintain the $bt$-tree dynamically using a union–find data structure
following \cite{westbrook1992maintaining}. Each bridge-component is
represented as a set in the union–find structure. When an edge is
inserted between two vertices belonging to different bridge-components,
the corresponding sets are merged using a union operation. Since each
edge insertion performs at most one union and a constant number of find
operations, the total cost of maintaining component membership is
$\mathcal{O}(|V_F|\alpha(|V_F|))$.

Importantly, the structure of the $bt$-tree changes only locally:
inserting an internal edge can merge at most two bridge-components and
affects only the path between them in the $bt$-tree. 

Each bridge-component carries a label of size at most three. When two
components merge, only their labels need to be updated, and since the
label size is bounded by a constant, this update takes constant time.
Furthermore, each edge insertion affects at most two labels, so label
maintenance over all insertions is linear. Thus, all dynamic updates to the $bt$-tree and labels over the course of
the algorithm require
$
\mathcal{O}(|V_F|(1+\alpha(|V_F|)))
$ time for each face $F$.

Since each face is processed independently and the sum of $|V_F|$ over
all faces is $\mathcal{O}(|V(G)|)$ in a plane graph, the total running
time is
$
\mathcal{O}(|V(G)|(1+\alpha(|V(G)|))).
$

For space, we store the $bt$-tree, union–find structure, and labels,
each of linear size. No step allocates more than linear auxiliary
memory, so the total space usage is $\mathcal{O}(|V(G)|)$.

The procedure \textsc{PA\_Bridgeconn} preserves the embedding and adds planar edges (\cref{lem:bridgeconnplanar}).
\end{proof}

%% file: Sections/3.Triconnforfixedembed.tex
\section{Triconnectivity Augmentation of Outerplanar Graphs with a fixed Embedding}\label{triconnproblem}

In this section, we turn to the main contribution of the paper and move to the domain of Vertex Connectivity Augmentation. Recall that in OTA-Fix, the input is a biconnected Outerplane graph $G $. We are required to find a minimum-size set of
edges $E'$ such that $G'=G + E'$ has a planar embedding that extends the original embedding.

\subparagraph*{Section and method overview} Our strategy for solving OTA-Fix is split across various subsections. In \Cref{subsec:prelimstreegraph,subsec:biconnplanetree} we recall the concept of \emph{treegraph} introduced by Kant \cite{kant1996augmenting} for outerplanar graphs and an algorithm that given a treegraph outputs a minimum set of edges that makes the treegraph biconnected such that the endpoint of every augmentation edge is a leaf vertex. Furthermore, we discuss the equivalences between the solution edges of the treegraph and solution edges of a biconnected outerplane graph. Finally, we discuss our novel contribution where we find an augmentation that is at most one edge away from optimal (\cref{almostopt}), and then adjust it to an optimal solution in \cref{optalg}. The high-level idea is to handle separately the case of an ‘unbalanced’ face with odd number of degree-2 vertices by adding one or two interior edges inside a face of the graph.

\subsection{Preliminaries}\label{subsec:prelimstreegraph}
When we consider the input graph to be embedded on a plane, all the faces are \emph{internal} except for one face which is \emph{external} and unbounded. In similar vein, the edges that we add to augment the graph (edges in the solution) are of two kinds. Either they lie in an internal face, or they lie in the outer face. In case of the latter, it is called an \emph{outer edge} otherwise, it is called an \emph{inner edge}. Note that the distinction is made before any edges are added to the graph. For a given pair of non-adjacent vertices $u,v$ in an internal  face of the graph, the augmented edge $uv$ can be an inner edge or an outer edge. 

Before we proceed further, let us observe a lower bound on the size of a solution for \textsc{OTA-Fix}.
\begin{observation}
     For a given 2-connected outerplane graph $G$ with $k$ degree 2 vertices in the graph, we require at least $\ceil{\frac{k}{2}}$ edges to 3-vertex connect the graph.
\end{observation}
Let us now look at a \emph{Tree} structure of the graph that was originally described in \cite{kant1996augmenting}.

\begin{definition}
    $T(G)$ or \emph{Tree of the graph} for an outerplanar plane graph $G$ is the interior dual of $G$ with one pendant attached for each degree-2 vertex in $G$, embedded according to the cyclic order inherited from $G$.

\end{definition}

Note that the interior dual for a biconnected outerplane graph is a tree. In the following, there is a change in notation; in the previous sections, we used $p$ to denote the number of pendants. Let $k$ denote the number of pendants for $T(G)$.

\begin{figure}
  \centering
  \begin{minipage}[t]{0.48\textwidth}
    \centering
    \input{Figs/tikzfigs/dualtree.tex}
    \smallskip

    \textbf{(a)}
    \label{fig:dualtree}
  \end{minipage}\hfill
  \begin{minipage}[t]{0.48\textwidth}
    \centering
    \input{Figs/tikzfigs/treegraph.tex}
    \smallskip

    \textbf{(b)}
    \label{fig:treegraph}
  \end{minipage}

  \caption{(a) The graph $G$ (black) and its dual graph (red), excluding the outer face.
(b) The interior dual graph with a pendant attached for each degree-2 vertex, inheriting the canonical embedding from the embedding of $G$.}
  \label{fig:dual-and-tree}
\end{figure}
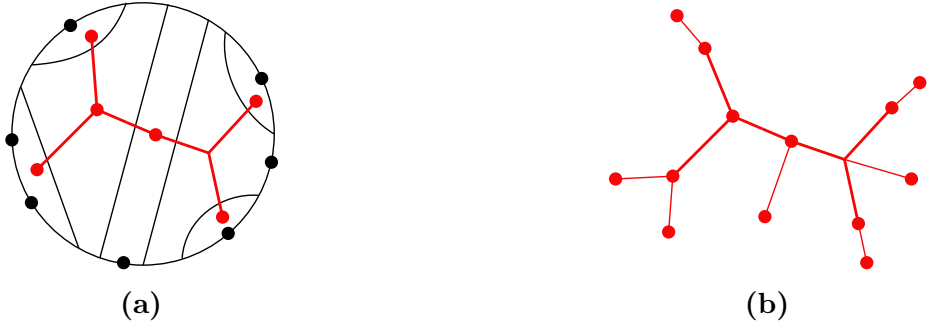

As an extension from previously used terminology, we define balanced graphs as follows based on balanced $bc$-trees

\begin{definition}
    An outerplane 2-connected graph $G$ is said to be balanced if $bc(T(G))$ is balanced.
\end{definition}
\cref{fig:dual-and-tree} provides an example of an outerplane graph and its corresponding treegraph.

\subsection{Plane Biconnectivity Augmentation for a Balanced Tree}\label{subsec:biconnplanetree}

In this subsection, we discuss an algorithm that outputs given a balanced tree, a minimum set of edges that makes the graph 2-vertex-connected ensuring that the output satisfies the additional properties stated in \cref{lem_onlyleaves}. Note that the reason we discuss the augmentation of a plane tree is revealed when it helps us with the augmentation of an outerplane biconnected graph as used in \Cref{almostopt,optalg}. 

The algorithm description is as follows:

Algorithm {\textsc{TreeConn}}: 

Let $G$ be a balanced plane tree. Compute its block--cut tree $bc(G)$ and the
associated labels, and initialize $A:=\emptyset$. While more than one label
remains, select a label $l_1$ of maximum size and a neighboring label $l_2$ in
the cyclic order induced by the embedding. Let $p_1\in l_1$ and $p_2\in l_2$ be
pendant blocks that are consecutive in this cyclic order. Insert an edge $e$
between simple vertices (non cut-vertices) of $p_1$ and $p_2$, update $G\leftarrow G+e$, and add
$e$ to $A$.

When exactly one label $l$ remains, if $|l|=2$ then insert a single edge between
simple vertices of the two pendant blocks. If $|l|=3$ with pendant blocks
$p_1,p_2,p_3$, choose simple vertices $v_i\in p_i$ that are leaves of the
original tree $G$, and add two edges $e_1=v_1v_2$ and $e_2=v_2v_3$ such that the
endpoints of each edge do not share a common parent in $G$ (i.e., neither
$\{v_1,v_2\}$ nor $\{v_2,v_3\}$ are siblings). The algorithm outputs the set $A$
of all inserted edges.

We recall a lemma that provides a sufficient condition when a $c$-node of high degree should be a parent of a label.

\begin{lemma}[{\cite[Lemma~2]{gutwenger2009planar}}]\label{impolemma1}
	If a $bc$-tree $\mathcal{B}$ is not balanced, its massive $c$-node $c^*$ must be the parent of a label, say $\ell_0$. Furthermore, $\ell_0$ is the unique maximum label in $\mathcal{B}$.
\end{lemma}

With this, we can proceed to prove the optimality of the algorithm with extra conditions.

\begin{restatable}{lemma}{balancedtreemin}\label{lem_balancedtreemin}
	Given a balanced plane tree $G$ with $p$ leaves, \textsc{TreeConn} computes an optimal augmentation of size $\lceil p/2\rceil$ that makes $G$ biconnected. Furthermore, every augmentation edge is incident only on leaves of $G$, and no augmentation edge joins two leaves with a common parent.
\end{restatable}

\input{Sections/proofs/balancedtreemin}
The arguments for the time and space finalises the analysis.

\begin{restatable}{lemma}{treeconnruntime}\label{lem_nearlintime}
	\textsc{TreeConn} can be implemented to run in $\mathcal{O}(1 +  \alpha(|V(G)|)|V(G)|)$ time and linear space.
\end{restatable}
	\begin{proof}
We first compute the block--cut tree $bc(G)$ of the input tree $G$, which can be
done in linear time. Using the leaf connectivity condition
\cite{gutwenger2009planar}, we partition the pendants into labels in overall
linear time.

When an augmentation edge is added, at most four labels are affected.
Updating label membership and sizes therefore takes constant time per edge
insertion.

We maintain the dynamic merging of components in $bc(G)$ using a union--find
structure \cite{westbrook1992maintaining}. Each edge insertion triggers a
constant number of \textsc{find} and \textsc{union} operations, yielding a total
time of $O(n\cdot\alpha(n))$ for all updates, where $\alpha(\cdot)$ is the inverse
Ackermann function.

Since the algorithm processes each face independently and the total size of
all facegraphs is $O(n)$, the overall running time is
$O(n(1+\alpha(n)))$.

For space, we store the $bc$-trees, union--find structure, and label data each requiring linear space.

\end{proof}

\textsc{TreeConn} is an augmentation procedure that provides a minimum set of edges satisfying some special properties. This is because we would eventually like the set of edges output by \textsc{TreeConn} to be used for triconnectivity augmentation of biconnected outerplane graphs. The specific properties are as stated in the lemma below:

\begin{lemma}\label{lem_onlyleaves}
For a balanced plane tree $T$ with $k$ leaves (that is not $K_2$ or $K_3$), we can
compute in time $O(1+\alpha(|V|)\,|V|)$ and linear space an optimal augmentation
$A_T$ that biconnects $T$ with the following properties:
\begin{itemize}
  \item $|A_T| = \pendant{k}$.
  \item Every edge in $A_T$ is incident to a leaf of $T$.
  \item No edge in $A_T$ joins two leaves of $T$ that share the same parent.
\end{itemize}
\end{lemma}
\begin{proof}
	This is a direct consequence of the \cref{lem_balancedtreemin,lem_nearlintime}.
\end{proof}

Now we make the transition from the biconnectivity of a tree to the triconnectivity of an outerplane graph. For this reason, we need to show that the edges obtained while making $T(G)$ biconnected can be translated and used as the edges required to make $G$ triconnected. For the case when $G$ is balanced, the transition is direct.

\begin{restatable}{lemma}{propbalancedopt}\label{propbalancedopt}
	For a given balanced outerplane graph $G$, we can compute $A_{\mbox{tri}}$ in time $O(1+\alpha(|V|)\,|V|)$ and linear space such that $G+A_{\mbox{tri}}$ is triconnected and $|A_{\mbox{tri}}| = \pendant{k}$.
\end{restatable}
\input{Sections/proofs/propbalancedopt} 

The implications of biconnecting the plane tree as shown has further consequences which shall be highlighted in \Cref{optalg}. We shall now proceed to the algorithmic portion for solving \textsc{OTA-Fix}.

%% file: Figs/tikzfigs/dualtree.tex
\scalebox{0.7}{

\tikzset{every picture/.style={line width=0.75pt}} 

\begin{tikzpicture}[x=0.75pt,y=0.75pt,yscale=-1,xscale=1]
	
	\draw   (124.5,139.5) .. controls (124.5,87.59) and (166.59,45.5) .. (218.5,45.5) .. controls (270.41,45.5) and (312.5,87.59) .. (312.5,139.5) .. controls (312.5,191.41) and (270.41,233.5) .. (218.5,233.5) .. controls (166.59,233.5) and (124.5,191.41) .. (124.5,139.5) -- cycle ;
	\draw  [fill={rgb, 255:red, 0; green, 0; blue, 0 }  ,fill opacity=1 ] (120.26,143.74) .. controls (120.26,141.4) and (122.16,139.5) .. (124.5,139.5) .. controls (126.84,139.5) and (128.74,141.4) .. (128.74,143.74) .. controls (128.74,146.08) and (126.84,147.98) .. (124.5,147.98) .. controls (122.16,147.98) and (120.26,146.08) .. (120.26,143.74) -- cycle ;
	\draw    (131,106) -- (172.5,221.48) ;
	\draw    (236.5,46.48) -- (187.5,228.48) ;
	\draw    (265.5,57.48) -- (218.5,233.5) ;
	\draw    (139,90) .. controls (176.5,87.48) and (195.5,74.48) .. (206.5,45.48) ;
	\draw    (277.5,66.48) .. controls (273.5,93.48) and (288.5,125.48) .. (312.5,139.5) ;
	\draw    (301.5,183.48) .. controls (274.5,181.48) and (248.5,200.48) .. (246.5,228.5) ;
	\draw  [fill={rgb, 255:red, 0; green, 0; blue, 0 }  ,fill opacity=1 ] (134.26,188.74) .. controls (134.26,186.4) and (136.16,184.5) .. (138.5,184.5) .. controls (140.84,184.5) and (142.74,186.4) .. (142.74,188.74) .. controls (142.74,191.08) and (140.84,192.98) .. (138.5,192.98) .. controls (136.16,192.98) and (134.26,191.08) .. (134.26,188.74) -- cycle ;
	\draw  [fill={rgb, 255:red, 0; green, 0; blue, 0 }  ,fill opacity=1 ] (200.26,231.74) .. controls (200.26,229.4) and (202.16,227.5) .. (204.5,227.5) .. controls (206.84,227.5) and (208.74,229.4) .. (208.74,231.74) .. controls (208.74,234.08) and (206.84,235.98) .. (204.5,235.98) .. controls (202.16,235.98) and (200.26,234.08) .. (200.26,231.74) -- cycle ;
	\draw  [fill={rgb, 255:red, 0; green, 0; blue, 0 }  ,fill opacity=1 ] (162.26,61.74) .. controls (162.26,59.4) and (164.16,57.5) .. (166.5,57.5) .. controls (168.84,57.5) and (170.74,59.4) .. (170.74,61.74) .. controls (170.74,64.08) and (168.84,65.98) .. (166.5,65.98) .. controls (164.16,65.98) and (162.26,64.08) .. (162.26,61.74) -- cycle ;
	\draw  [fill={rgb, 255:red, 0; green, 0; blue, 0 }  ,fill opacity=1 ] (299.26,99.74) .. controls (299.26,97.4) and (301.16,95.5) .. (303.5,95.5) .. controls (305.84,95.5) and (307.74,97.4) .. (307.74,99.74) .. controls (307.74,102.08) and (305.84,103.98) .. (303.5,103.98) .. controls (301.16,103.98) and (299.26,102.08) .. (299.26,99.74) -- cycle ;
	\draw  [fill={rgb, 255:red, 0; green, 0; blue, 0 }  ,fill opacity=1 ] (306.26,159.74) .. controls (306.26,157.4) and (308.16,155.5) .. (310.5,155.5) .. controls (312.84,155.5) and (314.74,157.4) .. (314.74,159.74) .. controls (314.74,162.08) and (312.84,163.98) .. (310.5,163.98) .. controls (308.16,163.98) and (306.26,162.08) .. (306.26,159.74) -- cycle ;
	\draw  [fill={rgb, 255:red, 0; green, 0; blue, 0 }  ,fill opacity=1 ] (275.26,210.74) .. controls (275.26,208.4) and (277.16,206.5) .. (279.5,206.5) .. controls (281.84,206.5) and (283.74,208.4) .. (283.74,210.74) .. controls (283.74,213.08) and (281.84,214.98) .. (279.5,214.98) .. controls (277.16,214.98) and (275.26,213.08) .. (275.26,210.74) -- cycle ;
	\draw [color={rgb, 255:red, 246; green, 6; blue, 6 }  ,draw opacity=1 ][line width=1.5]    (181.5,69.49) -- (185.5,122.08) ;
	\draw [color={rgb, 255:red, 246; green, 6; blue, 6 }  ,draw opacity=1 ][line width=1.5]    (185.5,122.08) -- (206.5,131.08) ;
	\draw [color={rgb, 255:red, 246; green, 6; blue, 6 }  ,draw opacity=1 ][line width=1.5]    (206.5,131.08) -- (227.5,140.08) ;
	\draw [color={rgb, 255:red, 246; green, 6; blue, 6 }  ,draw opacity=1 ][line width=1.5]    (227.5,140.08) -- (265.5,153.08) -- (275.5,199.08) ;
	\draw [color={rgb, 255:red, 246; green, 6; blue, 6 }  ,draw opacity=1 ][line width=1.5]    (185.5,122.08) -- (142.5,165.08) ;
	\draw [color={rgb, 255:red, 246; green, 6; blue, 6 }  ,draw opacity=1 ][line width=1.5]    (299.5,116.08) -- (265.5,153.08) ;
	\draw  [color={rgb, 255:red, 246; green, 6; blue, 6 }  ,draw opacity=1 ][fill={rgb, 255:red, 246; green, 6; blue, 6 }  ,fill opacity=1 ] (138.26,165.08) .. controls (138.26,162.74) and (140.16,160.84) .. (142.5,160.84) .. controls (144.84,160.84) and (146.74,162.74) .. (146.74,165.08) .. controls (146.74,167.42) and (144.84,169.32) .. (142.5,169.32) .. controls (140.16,169.32) and (138.26,167.42) .. (138.26,165.08) -- cycle ;
	\draw  [color={rgb, 255:red, 246; green, 6; blue, 6 }  ,draw opacity=1 ][fill={rgb, 255:red, 246; green, 6; blue, 6 }  ,fill opacity=1 ] (181.26,122.08) .. controls (181.26,119.74) and (183.16,117.84) .. (185.5,117.84) .. controls (187.84,117.84) and (189.74,119.74) .. (189.74,122.08) .. controls (189.74,124.42) and (187.84,126.32) .. (185.5,126.32) .. controls (183.16,126.32) and (181.26,124.42) .. (181.26,122.08) -- cycle ;
	\draw  [color={rgb, 255:red, 246; green, 6; blue, 6 }  ,draw opacity=1 ][fill={rgb, 255:red, 246; green, 6; blue, 6 }  ,fill opacity=1 ] (223.26,140.08) .. controls (223.26,137.74) and (225.16,135.84) .. (227.5,135.84) .. controls (229.84,135.84) and (231.74,137.74) .. (231.74,140.08) .. controls (231.74,142.42) and (229.84,144.32) .. (227.5,144.32) .. controls (225.16,144.32) and (223.26,142.42) .. (223.26,140.08) -- cycle ;
	\draw  [color={rgb, 255:red, 246; green, 6; blue, 6 }  ,draw opacity=1 ][fill={rgb, 255:red, 246; green, 6; blue, 6 }  ,fill opacity=1 ] (177.26,69.49) .. controls (177.26,67.15) and (179.16,65.25) .. (181.5,65.25) .. controls (183.84,65.25) and (185.74,67.15) .. (185.74,69.49) .. controls (185.74,71.84) and (183.84,73.73) .. (181.5,73.73) .. controls (179.16,73.73) and (177.26,71.84) .. (177.26,69.49) -- cycle ;
	\draw  [color={rgb, 255:red, 246; green, 6; blue, 6 }  ,draw opacity=1 ][fill={rgb, 255:red, 246; green, 6; blue, 6 }  ,fill opacity=1 ] (271.26,199.08) .. controls (271.26,196.74) and (273.16,194.84) .. (275.5,194.84) .. controls (277.84,194.84) and (279.74,196.74) .. (279.74,199.08) .. controls (279.74,201.42) and (277.84,203.32) .. (275.5,203.32) .. controls (273.16,203.32) and (271.26,201.42) .. (271.26,199.08) -- cycle ;
	\draw  [color={rgb, 255:red, 246; green, 6; blue, 6 }  ,draw opacity=1 ][fill={rgb, 255:red, 246; green, 6; blue, 6 }  ,fill opacity=1 ] (295.26,116.08) .. controls (295.26,113.74) and (297.16,111.84) .. (299.5,111.84) .. controls (301.84,111.84) and (303.74,113.74) .. (303.74,116.08) .. controls (303.74,118.42) and (301.84,120.32) .. (299.5,120.32) .. controls (297.16,120.32) and (295.26,118.42) .. (295.26,116.08) -- cycle ;

\end{tikzpicture}
}

%% file: Figs/tikzfigs/treegraph.tex
\scalebox{0.7}{

\tikzset{every picture/.style={line width=0.75pt}} 

\begin{tikzpicture}[x=0.75pt,y=0.75pt,yscale=-1,xscale=1]
	
	\draw [color={rgb, 255:red, 246; green, 6; blue, 6 }  ,draw opacity=1 ][line width=1.5]    (159.5,97.48) -- (179.5,146.08) ;
	\draw [color={rgb, 255:red, 246; green, 6; blue, 6 }  ,draw opacity=1 ][line width=1.5]    (179.5,146.08) -- (200.5,155.08) ;
	\draw [color={rgb, 255:red, 246; green, 6; blue, 6 }  ,draw opacity=1 ][line width=1.5]    (200.5,155.08) -- (221.5,164.08) ;
	\draw [color={rgb, 255:red, 246; green, 6; blue, 6 }  ,draw opacity=1 ][line width=1.5]    (221.5,164.08) -- (259.5,177.08) -- (269.5,223.08) ;
	\draw [color={rgb, 255:red, 246; green, 6; blue, 6 }  ,draw opacity=1 ][line width=1.5]    (179.5,146.08) -- (136.5,189.08) ;
	\draw [color={rgb, 255:red, 246; green, 6; blue, 6 }  ,draw opacity=1 ][line width=1.5]    (293.5,140.08) -- (259.5,177.08) ;
	\draw  [color={rgb, 255:red, 246; green, 6; blue, 6 }  ,draw opacity=1 ][fill={rgb, 255:red, 246; green, 6; blue, 6 }  ,fill opacity=1 ] (175.26,146.08) .. controls (175.26,143.74) and (177.16,141.84) .. (179.5,141.84) .. controls (181.84,141.84) and (183.74,143.74) .. (183.74,146.08) .. controls (183.74,148.42) and (181.84,150.32) .. (179.5,150.32) .. controls (177.16,150.32) and (175.26,148.42) .. (175.26,146.08) -- cycle ;
	\draw  [color={rgb, 255:red, 246; green, 6; blue, 6 }  ,draw opacity=1 ][fill={rgb, 255:red, 246; green, 6; blue, 6 }  ,fill opacity=1 ] (217.26,164.08) .. controls (217.26,161.74) and (219.16,159.84) .. (221.5,159.84) .. controls (223.84,159.84) and (225.74,161.74) .. (225.74,164.08) .. controls (225.74,166.42) and (223.84,168.32) .. (221.5,168.32) .. controls (219.16,168.32) and (217.26,166.42) .. (217.26,164.08) -- cycle ;
	\draw  [color={rgb, 255:red, 246; green, 6; blue, 6 }  ,draw opacity=1 ][fill={rgb, 255:red, 246; green, 6; blue, 6 }  ,fill opacity=1 ] (132.26,189.08) .. controls (132.26,186.74) and (134.16,184.84) .. (136.5,184.84) .. controls (138.84,184.84) and (140.74,186.74) .. (140.74,189.08) .. controls (140.74,191.42) and (138.84,193.32) .. (136.5,193.32) .. controls (134.16,193.32) and (132.26,191.42) .. (132.26,189.08) -- cycle ;
	\draw  [color={rgb, 255:red, 246; green, 6; blue, 6 }  ,draw opacity=1 ][fill={rgb, 255:red, 246; green, 6; blue, 6 }  ,fill opacity=1 ] (289.26,140.08) .. controls (289.26,137.74) and (291.16,135.84) .. (293.5,135.84) .. controls (295.84,135.84) and (297.74,137.74) .. (297.74,140.08) .. controls (297.74,142.42) and (295.84,144.32) .. (293.5,144.32) .. controls (291.16,144.32) and (289.26,142.42) .. (289.26,140.08) -- cycle ;
	\draw  [color={rgb, 255:red, 246; green, 6; blue, 6 }  ,draw opacity=1 ][fill={rgb, 255:red, 246; green, 6; blue, 6 }  ,fill opacity=1 ] (265.26,223.08) .. controls (265.26,220.74) and (267.16,218.84) .. (269.5,218.84) .. controls (271.84,218.84) and (273.74,220.74) .. (273.74,223.08) .. controls (273.74,225.42) and (271.84,227.32) .. (269.5,227.32) .. controls (267.16,227.32) and (265.26,225.42) .. (265.26,223.08) -- cycle ;
	\draw  [color={rgb, 255:red, 246; green, 6; blue, 6 }  ,draw opacity=1 ][fill={rgb, 255:red, 246; green, 6; blue, 6 }  ,fill opacity=1 ] (155.26,97.48) .. controls (155.26,95.14) and (157.16,93.24) .. (159.5,93.24) .. controls (161.84,93.24) and (163.74,95.14) .. (163.74,97.48) .. controls (163.74,99.82) and (161.84,101.72) .. (159.5,101.72) .. controls (157.16,101.72) and (155.26,99.82) .. (155.26,97.48) -- cycle ;
	\draw [color={rgb, 255:red, 246; green, 6; blue, 6 }  ,draw opacity=1 ][line width=0.75]    (221.5,164.08) -- (202.5,218.09) ;
	\draw [color={rgb, 255:red, 246; green, 6; blue, 6 }  ,draw opacity=1 ][line width=0.75]    (136.5,189.08) -- (95.5,191.09) ;
	\draw [color={rgb, 255:red, 246; green, 6; blue, 6 }  ,draw opacity=1 ][line width=0.75]    (136.5,189.08) -- (133.5,229.09) ;
	\draw [color={rgb, 255:red, 246; green, 6; blue, 6 }  ,draw opacity=1 ][line width=0.75]    (139.5,74.09) -- (159.5,97.48) ;
	\draw [color={rgb, 255:red, 246; green, 6; blue, 6 }  ,draw opacity=1 ][line width=0.75]    (313.5,122.09) -- (293.5,140.08) ;
	\draw [color={rgb, 255:red, 246; green, 6; blue, 6 }  ,draw opacity=1 ][line width=0.75]    (259.5,177.08) -- (307.5,191.09) ;
	\draw [color={rgb, 255:red, 246; green, 6; blue, 6 }  ,draw opacity=1 ][line width=0.75]    (269.5,223.08) -- (275.5,251.09) ;
	\draw  [color={rgb, 255:red, 246; green, 6; blue, 6 }  ,draw opacity=1 ][fill={rgb, 255:red, 246; green, 6; blue, 6 }  ,fill opacity=1 ] (91.26,191.09) .. controls (91.26,188.75) and (93.16,186.85) .. (95.5,186.85) .. controls (97.84,186.85) and (99.74,188.75) .. (99.74,191.09) .. controls (99.74,193.44) and (97.84,195.33) .. (95.5,195.33) .. controls (93.16,195.33) and (91.26,193.44) .. (91.26,191.09) -- cycle ;
	\draw  [color={rgb, 255:red, 246; green, 6; blue, 6 }  ,draw opacity=1 ][fill={rgb, 255:red, 246; green, 6; blue, 6 }  ,fill opacity=1 ] (129.26,229.09) .. controls (129.26,226.75) and (131.16,224.85) .. (133.5,224.85) .. controls (135.84,224.85) and (137.74,226.75) .. (137.74,229.09) .. controls (137.74,231.44) and (135.84,233.33) .. (133.5,233.33) .. controls (131.16,233.33) and (129.26,231.44) .. (129.26,229.09) -- cycle ;
	\draw  [color={rgb, 255:red, 246; green, 6; blue, 6 }  ,draw opacity=1 ][fill={rgb, 255:red, 246; green, 6; blue, 6 }  ,fill opacity=1 ] (198.26,218.09) .. controls (198.26,215.75) and (200.16,213.85) .. (202.5,213.85) .. controls (204.84,213.85) and (206.74,215.75) .. (206.74,218.09) .. controls (206.74,220.44) and (204.84,222.33) .. (202.5,222.33) .. controls (200.16,222.33) and (198.26,220.44) .. (198.26,218.09) -- cycle ;
	\draw  [color={rgb, 255:red, 246; green, 6; blue, 6 }  ,draw opacity=1 ][fill={rgb, 255:red, 246; green, 6; blue, 6 }  ,fill opacity=1 ] (135.26,74.09) .. controls (135.26,71.75) and (137.16,69.85) .. (139.5,69.85) .. controls (141.84,69.85) and (143.74,71.75) .. (143.74,74.09) .. controls (143.74,76.44) and (141.84,78.33) .. (139.5,78.33) .. controls (137.16,78.33) and (135.26,76.44) .. (135.26,74.09) -- cycle ;
	\draw  [color={rgb, 255:red, 246; green, 6; blue, 6 }  ,draw opacity=1 ][fill={rgb, 255:red, 246; green, 6; blue, 6 }  ,fill opacity=1 ] (271.26,251.09) .. controls (271.26,248.75) and (273.16,246.85) .. (275.5,246.85) .. controls (277.84,246.85) and (279.74,248.75) .. (279.74,251.09) .. controls (279.74,253.44) and (277.84,255.33) .. (275.5,255.33) .. controls (273.16,255.33) and (271.26,253.44) .. (271.26,251.09) -- cycle ;
	\draw  [color={rgb, 255:red, 246; green, 6; blue, 6 }  ,draw opacity=1 ][fill={rgb, 255:red, 246; green, 6; blue, 6 }  ,fill opacity=1 ] (303.26,191.09) .. controls (303.26,188.75) and (305.16,186.85) .. (307.5,186.85) .. controls (309.84,186.85) and (311.74,188.75) .. (311.74,191.09) .. controls (311.74,193.44) and (309.84,195.33) .. (307.5,195.33) .. controls (305.16,195.33) and (303.26,193.44) .. (303.26,191.09) -- cycle ;
	\draw  [color={rgb, 255:red, 246; green, 6; blue, 6 }  ,draw opacity=1 ][fill={rgb, 255:red, 246; green, 6; blue, 6 }  ,fill opacity=1 ] (309.26,122.09) .. controls (309.26,119.75) and (311.16,117.85) .. (313.5,117.85) .. controls (315.84,117.85) and (317.74,119.75) .. (317.74,122.09) .. controls (317.74,124.44) and (315.84,126.33) .. (313.5,126.33) .. controls (311.16,126.33) and (309.26,124.44) .. (309.26,122.09) -- cycle ;

\end{tikzpicture}
}

%% file: Sections/proofs/balancedtreemin.tex
\begin{proof}
	
	We would like to proceed using induction.  The $bc$-tree is balanced. For the base case of 2 pendants, $\deg(c^*)$ will be 2. Here, we augment the graph using one edge and we are done. If $p=3$, then if it's from 1 label, we can add 2 edges inside the same label and complete. However, if we have 3 labels, we still only augment 2 edges from the algorithm. In both cases we've used $\ceil{\frac{p}{2}}$ edges. 
	
	\textbf{Induction hypothesis}: For $p' < p$ pendants we can augment using $\ceil{\frac{p'}{2}}$ edges

	Since $p \ge 4$, there exists at least 2 labels (massive $c$-node otherwise). Hence, the algorithm adds a new edge between 2 different labels. Let $l_1$ and $l_2$ be the labels. Let $l_1$ be a label of maximum size.
	
	\vspace{2mm}
	Assume on the contrary that addition of the edge $e$ creates a massive $c$-node $c^*$. Let $\deg(c^*)_{bef}$ and $\deg(c^*)_{aft}$ refer to the degree of $c^*$ before and after the edge $e$ is added respectively. 
    
    So, $\deg(c^*)_{aft} \ge \ceil{\frac{p'}{2}} + 2$. Assume that $c^*$ was in the cycle created by the addition of the edge $e$. This would mean that $\deg(c^*)_{bef} = \deg(c^*)_{aft} + 1 \ge \ceil{\frac{p'}{2}} + 3 = \ceil{\frac{p}{2}} + 2$. However, this means that $c^*$ was a massive node, contradicting the fact that there was no massive node.
	
	Hence, assume that $c^*$ does not lie in the cycle created by the edge added. We observe that $\deg(c^*)_{bef} = \ceil{\frac{p}{2}} + 1 = d$, say. Indeed, if $\deg(c^*)_{bef}$ were lesser, $c^*$ can't become massive after the edge addition.
	
	\vspace{2mm}
	From \cref{impolemma1}, we see that $c^*$-label exists with maximum size (say $j$). Call this label $l_0$. However we know that $c^*$ is not in the cycle formed by the new edge. This means that The edge is added between two other labels $l_1$ and $l_2$. However we required the edge to be added to maximum size label. Hence $l_2$ (WLOG)
	has size at least $j$. Remember that $\deg(c^*)_{bef} = d$ and the size $j$ label accounts to $j$ pendants. But the other $d-j-1$ neighbours account for at least 2 pendants each since they are not part of the $c^*$-label. The remaining 1 neighbor accounts for the labels $l_1$ and $l_2$ with at least $1$ and $j$ pendants respectively. Now we count the number of pendants ($p$)
	\begin{align*}
		p &\ge 2j + 2(d-j-1) + 1\\
		&=   2d-1\\
		&= 2\left(\left\lceil\frac{p}{2}\right\rceil + 1\right) - 1\\\
		&\ge p+1
	\end{align*}
	However, this is a contradiction. Hence, we can conclude that no massive $c$-node arises after the edge addition. Hence we can go through with the induction hypothesis to augment $G'$ (the resultant graph) with $\ceil{\frac{p'}{2}}$ edges and our original graph with $\ceil{\frac{p}{2}}$ edges.
	
	The algorithm makes sure that edges are not added between leaves with the same parent. Indeed, this is not possible when adding an edge between different labels. Assume that the algorithm has reached a point where there is exactly one label, the size of the label $l$ is either $l=3$ or $l=2$. If $l=3$, since $G$ is not $K_3$,  for the 3 simple vertices $v_1,v_2,v_3$ in the respective bridgeblocks, we know that at most 2 of them share the same parent. Indeed, if all 3 vertices shared the same parent, then we would require another edge from the parent to be connected to the rest of the graph. However, this is not possible since the label has size 3. It is not possible that the entire graph is the 3 pendants and the parent since this is $K_3$ and we forbade it. Hence, with at most two vertices sharing the same parent, we can make sure to order them such that neither $v_1,v_2$ nor $v_2,v_3$ share parents. 
	
	We also posit that initially, the pendant blocks comprise the $p$ degree 1 vertices as their only simple vertices (and their incident edge) in the tree $G$. Every time we add an edge between two different labels, the number of pendants go down by 2 and no new pendant block arises. Therefore, we are adding edges between leaves of the tree $G$. In case the graph only has 1 label, the construction still shows that we handpick the leaves from the pendants to add edges. Hence the set of augmentation edges only connects leaves of the graph $G$.
	\vspace{2mm}
\end{proof}

%% file: Sections/proofs/propbalancedopt.tex
\begin{proof}
	Consider the tree of the graph $G$, $T(G)$. By using Lemma~\ref{lem_onlyleaves}, we obtain $A_T$ such that $|A_T| = \ceil{\frac{k}{2}}$ and $T(G) + A_T$ is biconnected. For every edge between 2 leaves in $A_T$, identify an edge between the corresponding degree 2 vertices in $G$. Let $A_{\mbox{tri}}$ represent these set of edges. It is important to notice that there are no crossings and the planarity is preserved in the fixed embedding since $A_{\mbox{tri}}$ doesn't have any crossings. We claim that $G+A_{\mbox{tri}}$ is triconnected. 
	
	Let $v_1, v_2, \hdots, v_n$ represent the set of vertices in $G$ in cyclic order. Let $P_{v_i, v_j}$ represent the vertices $v_i, v_{i+1}, \hdots, v_j$.\footnote{when we say $v_i$ we mean $v_{i\mod n}$ since the ordering is cyclic}

	Consider two vertices $v_i$ and $v_j$ arbitrarily. We need to show that $G-\{v_i,v_j\}$ is still connected. We'll consider two cases.
	
	\begin{enumerate}[wide, labelindent=0pt,label=Case~\arabic*.]
		\item\label{itm:diffFace} {\bf $v_i$ and $v_j$ belong to different internal faces:} In this case, on removing $v_i$ and $v_j$, the path $P_{v_{i+1}, v_{j-1}}$ and $P_{v_{j+1}, v_{i-1}}$ still stays. But since $v_i$ and $v_j$ belong to different faces, there is a chord $e$ between $P_{v_{i+1}, v_{j-1}}$ and $P_{v_{j+1}, v_{i-1}}$. Hence the remaining vertices are still all connected.
		\item\label{itm:sameFace}  {\bf $v_i$ and $v_j$ belong to an internal face:} Consider an internal face with $v_i$ and $v_j$. Removing $v_i$ and $v_j$ will give us two paths as before $P_{v_{i+1}, v_{j-1}}$ and $P_{v_{j+1}, v_{i-1}}$. Say $v_i$ and $v_j$ only have degree 2 vertices between them in the path. We can assume that $v_i$ and $v_j$ are not adjacent because $G$ remains connected after their removal otherwise. Due to this assumption, let $v$ be a degree 2 vertex in the arc $P_{v_{i+1}, v_{j-1}}$ in the face $F$. From \cref{lem_onlyleaves}, we know that in $A_T$, the vertex $v$ is connected to a some vertex say $u$ that doesn't share the same parent ($F$). Hence, $v$ is connected to a vertex $u$ from $P_{v_{j+1}, v_{i-1}}$. This means that $G - \{v_i,v_j\}$ is connected.
		
		Otherwise, there are degree 3 vertices in  $P_{v_{i+1}, v_{j-1}}$ and $P_{v_{j+1}, v_{i-1}}$ that are adjacent to other faces . Let vertices in each path can be grouped into faces (identify a face $F \in \mathcal{F}$ even if one vertex in the path $P$ is incident on $F$). Let $\mathcal{F}_1 \text{ and } \mathcal{F}_2$ represent the grouping of the faces from the paths $P_{v_{i+1}, v_{j-1}}$ and $P_{v_{j+1}, v_{i-1}}$ respectively. Notice that $\mathcal{F}_1 \text{ and } \mathcal{F}_2$ are not empty from the very assumption of the case. Let $V(\mathcal{F}_i)$  represent the veertices corresponding to the faces in $\mathcal{F}_i$ for $i \in \{1,2\}$. Since $T(G) - v_F$ is connected, there exists an outer edge $e$ from some vertex $v_{F_1} \in \mathcal{F}_1$ to some $v_{F_2} \in \mathcal{F}_2$ that connects degree 2 vertices from $\mathcal{F}_1$ and $\mathcal{F}_2$.

		WLOG, we have $v_a \in P_{v_{i+1}, v_{j-1}}$ and $v_b \in P_{v_{j+1}, v_{i-1}}$. $P_{v_{i+1}, v_{j-1}}$ and $P_{v_{j+1}, v_{i-1}}$ the paths are connected by $e$. Hence, the graph remains connected after the removal of $v_1$ and $v_j$. 
		
		It is also important to notice that the planarity criteria is satisfied and no new crossings are formed. This is because the edges $A_{\mbox{tri}}$ exactly correspond (in embedding) to the edges $A_T$ and they have no crossings.
		
	\end{enumerate}
\end{proof}

%% file: Sections/4.almostopt.tex
\subsection{An Additive (+1)-Approximation Algorithm for Triconnectivity Augmentation of Outerplane Graphs}\label{almostopt}

In this subsection, we present an algorithm for augmenting biconnected outerplane graphs that is almost optimal i.e, it might be off by an additional term of 1. The result is stated in the following lemma:

\begin{restatable}{lemma}{almostopt}\label{theorem_almostopt}
    \textsc{TriConn} triconnects an input biconnected outerplane graph $G=(V,E)$ with at most $\pendant{k}+1$ edges (where $k$ represents the number of degree-2 vertices in $V$) and runs in time $\mathcal{O}(|V|(1+\alpha(|V|)))$
\end{restatable}

We now briefly describe the structure of the algorithm used in this subsection. 
\paragraph*{Algorithmic overview.}

The algorithm \textsc{TriConn} computes an augmentation that is at most one edge
larger than optimal. For balanced instances, it invokes
\textsc{OuterPlanarConnect}, which is optimal. For unbalanced instances, it
adds a carefully chosen edge inside a maximum-degree face so as to
balance the graph, and then applies OuterPlanarConnect to the resulting
instance.

Before we proceed further, we provide essential preliminaries and structural insights.

\subsubsection*{Preliminaries and Structural Results}

Throughout the remainder of the section, we talk about internal faces in the graph and make a distinction between faces that are large and otherwise. Hence, it is essential and natural to introduce some new terminology.

\begin{definition}
	For an Outerplane graph $G$ with an internal face $F$, the \emph{degree of a face} $F$ or $d(F)$ refers to the degree of the vertex $v_F$ representing the face in $T(G)$.
\end{definition}   
A face $F$ is said to be \emph{massive} if in $T(G)$ we have $d(v_F)-1 > \pendant{k}$.

We observe that if a face has very high degree, it is not possible for the other faces in the graph to have a high degree.

\begin{restatable}{lemma}{bigfacesmallface}\label{lem:bigfacesmallface}
	Let $G$ be an outerplane graph with $k$ degree-2 vertices. If $G$ contains a face $F$ that is massive, then every other face $F'$ satisfies $d(F') \le \pendant{k}$.
\end{restatable}

\begin{proof}
	Let $v_F$ be the node corresponding to a face $F$ with
$d(F) \ge \pendant{k}+2$. Suppose for contradiction that there exists
another face $F'$ with node $v_{F'}$ such that
$d(F') \ge \pendant{k}+1$.

Consider the unique path between $v_F$ and $v_{F'}$ in $T(G)$.
Remove the edges of this path. The remaining graph consists of two
subtrees rooted at $v_F$ and $v_{F'}$, respectively. Each pendant
belongs to exactly one of these subtrees or lies on the path.

The subtree rooted at $v_F$ contains at least $\pendant{k}+1$ pendants,
since $d(F) \ge \pendant{k}+2$ and at most one edge of the path is
incident to $v_F$. Similarly, the subtree rooted at $v_{F'}$ contains
at least $\pendant{k}$ pendants, since $d(F') \ge \pendant{k}+1$.

Thus the total number of pendants in $T(G)$ is at least
\[
\left(\pendant{k}+1\right) + \pendant{k} = k+1,
\]
contradicting the fact that $G$ has exactly $k$ degree-2 vertices,
and hence $T(G)$ has exactly $k$ pendants.
\end{proof}

Now, we look at a structural lemma that guarantees a scenario where an edge added inside a face $F$ balances the graph.

\begin{restatable}{lemma}{newmaxface}\label{lem_newmaxface}
	If $G$ is an unbalanced graph with a massive face $F$ such that adding $e$ inside $F$ leads to an old face $F'$ in $G (\neq F)$ becoming the maximum
	degree face in $G+e$, then $G+e$ is a balanced graph.
\end{restatable}

	\begin{proof}
		We know that the face $F$ has a degree of $d(F) > \pendant{k}+1$. We muster that $d(F') \le \pendant{k}$ for all other faces $F'\neq F$ in $G$ from arithmetic by counting the pendants from the fact that the total number of degree-2 vertices in $G$ is $k$. Once the edge $e$ is added, an old face $F^*$ is the maximum-degree face. But we know that $d(F^*) \le \pendant{k}$. That is, $d(F^*)-1 \le \pendant{k-2}$. $G+e$ has at least $k-2$ degree-2 vertices (since the addition of an edge removes two degree-2 vertices at most). Hence, none of the faces are massive, which means that the resultant graph is balanced.
	\end{proof}

\subsubsection*{Algorithm for (+1)-approximation} 

The algorithm proceeds as follows: for unbalanced graphs, we cleverly pick an edge in the maximum-degree face that splits the face into two new faces of similar size and makes use of the subroutine for balanced graphs highlighted in \cref{propbalancedopt}.

\begin{algorithm}[t]
\caption{\textsc{TriConn}}
\label{triconn}

\KwIn{Outerplane 2-connected graph $G$}
\KwOut{List of edges $E'$ such that $G(V, E \cup E')$ is planar and triconnected}

Compute $T(G)$, $k \gets \text{no. of leaves}$, and $d \gets \max_{v \in T(G)} d(v)$\;

\eIf{$\pendant{k} \ge d - 1$}{
    \Return{\textsc{OuterPlanarConnect}($G$)}\;
}{
    Let $v_F$ be the vertex corresponding to face $F$ with maximum-degree $d$\;
    Let $u$ be an arbitrary vertex in $F$ (choose it to be a degree $2$ vertex if it exists)\;
    Let $v$ be another vertex in $F$ such that $e = uv$ splits $F$ into two faces $F_1$ and $F_2$ with $|d(v_{F_1}) - d(v_{F_2})| \le 1$\;
    
    \Return $\textsc{OuterPlanarConnect}(G+e) \cup \{e\}$;
}
\end{algorithm}

\begin{algorithm}[H]
\caption{\textsc{OuterPlanarConnect}}
\label{outerplanarconnect}

\KwIn{Outerplane 2-connected balanced graph $G$}
\KwOut{List of edges $E'$ such that $G(V, E \cup E')$ is planar and triconnected}

$A \gets \textsc{TreeConn}(T(G))$\;
Let $A'$ represent the corresponding edge set in $G$ from \cref{propbalancedopt}\;
\Return{$A'$}\;

\end{algorithm}

Although this solution does not claim to be the optimal one, it is indeed optimal when the input graph is balanced. 

\begin{lemma}\label{lem_balancedopt}
For a balanced outerplane 2-connected graph $G$, \textsc{OuterPlanarConnect} computes an optimal triconnectivity augmentation in $\mathcal{O}(|V|(1+\alpha(|V|)))$ time.
\end{lemma}
\begin{proof}
If $G$ is balanced, \cref{propbalancedopt} implies that $\pendant{k}$ edges suffice and are necessary. The correctness and running time follow as well.
\end{proof}

We can finally prove the correctness of the solution with the desired Time Complexity.
\begin{proof}[Proof of \cref{theorem_almostopt}]
If $G$ is balanced, optimality and running time follow from \cref{lem_balancedopt}.  
Assume $G$ is unbalanced and let $F$ be a face with $d(F) > \lceil k/2 \rceil + 1$.  
By \cref{lem:bigfacesmallface}, all other faces $F' \ne F$ satisfy  
$d(F') \le \lceil k/2 \rceil$. Adding an internal edge $e$ that splits $F$ into faces
$F_1, F_2$ with $|d(F_1)-d(F_2)| \le 1$ ensures  
$d(F_i)-1 \le \lceil (k-a)/2 \rceil$ for $i\in\{1,2\}$, where $a\in\{0,1,2\}$ counts the
degree-2 vertices incident to $e$. Extrapolating from \cref{lem_newmaxface} and using the degree bounds on $F_1$ and $F_2$, $G+e$ is balanced. Applying
\textsc{OuterPlanarConnect} yields a triconnected graph using at most
$\lceil (k-a)/2 \rceil + 1 \le \lceil k/2 \rceil + 1$ edges.  
The running time follows from the linear-time split and \cref{lem_balancedopt}.
\end{proof}

Although \textsc{TreeConn} might use one edge more than the optimal, in multiple instances, it does provide the optimal solution. In fact, there are various instances where $\pendant{k}+1$ edges are required to make a graph triconnected.

%% file: Sections/5.opt.tex
\subsection{An Optimal Algorithm for Triconnectivity Augmentation of Outerplane Graphs}\label{optalg}
In this subsection, we finally provide an optimal algorithm for the triconnectivity augmentation problem on general biconnected outerplane graphs (OTA-Fix) building on \cref{almostopt,subsec:biconnplanetree}. 
Let us restate it for convenience.

\begin{restatable}{theorem}{triconnmainthm}\label{thm:triconnexact}
     There exists an algorithm that solves OTA-Fix by running in time $\mathcal{O}(|V|(1+\alpha(|V|)))$ time and linear space, augmenting $\mathrm{OPT}(G)\in\{\lceil k/2\rceil,\ \lceil k/2\rceil+1\}$ edges, where $\alpha$ denotes the inverse of the Ackermann function.
 \end{restatable}

\paragraph*{Proof strategy.}
The additive-$1$ algorithm of \cref{almostopt} already shows that every
instance can be augmented using at most
$\lceil k/2\rceil+1$ edges. To obtain an exact algorithm, we first
establish structural properties of optimal solutions. In particular,
we show that unbalanced instances require internal augmentation edges
when only $\lceil k/2\rceil$ edges are used in an optimal solution. Morever, we establish that all the internal edges can be added to the maximum-degree face (\cref{internaledgeoneface}). These structural
results allow us to distinguish instances whose optimum is
$\lceil k/2\rceil+1$ from those admitting a solution of size
$\lceil k/2\rceil$. The resulting algorithm either identifies a
balancing edge in the maximum-degree face or reduces the problem to a
final balancing case.

The algorithm proceeds in stages. We first handle the case
where the maximum-degree face contains at most four degree-2 vertices,
for which an optimal augmentation can be obtained directly. We then
consider several additional configurations that admit simple balancing
steps or immediately imply that
$\mathrm{OPT}(G)=\lceil k/2\rceil+1$. 

Finally, once these
cases have been eliminated, the remaining instance satisfies a well-defined set of entry conditions.
The algorithm first selects five well-spaced degree-2 vertices on the bounday of the maximum-degree face $F$. It then analyzes the facial segments determined by these vertices, and adds one or two internal edges to balance the graph. The resulting balanced instance is then completed optimally using
\textsc{OuterPlanarConnect}.

All edges inserted by the algorithm are drawn inside faces and are
non-crossing, so the embedding of $G$ is preserved throughout.

Before we proceed to the structural preliminaries required for this section, we record two consequences from our results in \Cref{subsec:biconnplanetree}.

\subsubsection*{Supporting lemmas for \Cref{optalg}}

We observed that \textsc{TreeConn} provides a minimum edge set that helps us move from biconnectivity of a plane balanced tree to the triconnectivity of an outerplane balanced graph.

The following lemma shows that whenever an optimal augmentation using only outer edges exists for \textsc{OTA-Fix}, we may assume, without loss of generality, that all augmentation edges are incident only on degree-2 vertices.

\begin{restatable}{lemma}{tritotrideg}\label{tritotrideg2}
	If an outerplane biconnected graph $G$ can be triconnected using $\pendant{k}$ outer edges $A_{opt}$ (and no internal edges), then there exists $A_{tri}$ such that $A_{tri}$ only connects degree-2 vertices in $G$ and $|A_{tri}| = \pendant{k}$ such that $G+A_{tri}$ is triconnected.
\end{restatable}
\input{Sections/proofs/tritotrideg}

It is important to observe what happens when the graph is unbalanced. The following lemma and corollary discuss this.

\begin{restatable}{lemma}{outerktwo}\label{outerk2}
    If an outerplane biconnected graph $G$ can be triconnected using $\pendant{k}$ outer edges (and no internal edges), then $G$ is a balanced graph.
\end{restatable}

\input{Sections/proofs/outerktwo}

From \Cref{degreegoverns}, unbalanced outerplane graphs requires at least 1 internal edge in any triconnectivity augmentation using only $\pendant{k}$ edges. In the upcoming section, we introduce how a minimum edge addition can make an unbalanced graph into a balanced one.
\begin{corollary}\label{degreegoverns}
    If $G$ is an unbalanced outerplane graph, then $>\pendant{k}$ outer edges are required to make $G$ 3-connected.
\end{corollary}


\subsubsection*{Structural Preliminaries and Auxiliary Results}

We introduce terminology and definitions that shall be used throughout the remainder of the section. For vertices $u,v$ on the boundary of $F$, let $P_{uv}$ denote the clockwise
boundary path from $u$ to $v$ (and $P_{vu}$ the complementary path).
Inserting the edge $uv$ splits $F$ into faces $F_1$ and $F_2$ bounded by
$P_{uv}\cup uv$ and $P_{vu}\cup uv$, respectively. Define
$t(u,v) := d(F_1)-1$ and $t(v,u) := d(F_2)-1$.

Two degree-2 vertices $u$ and $v$ on the boundary of a face $F$ are said to be \emph{consecutive} along $P_{uv}$ if all the interior vertices in $P_{uv}$ have degree more than 2.
We move to important structural insights.

\begin{restatable}{lemma}{internaledgeoneface}\label{internaledgeoneface}
There exists an optimal solution for OTA-Fix where all the internal edges are added inside the maximum-degree face.
\end{restatable}
\input{Sections/proofs/internaledgeoneface}

We can now show that for the case where the input graph $G$ has an odd number of degree-2 vertices, the resulting augmentation is already optimal.

\begin{restatable}{lemma}{oddoptimal}\label{lem_oddoptimal}  
	If $k$ is odd, then \textsc{TriConn} returns an optimal solution to OTA-Fix.
\end{restatable}
\input{Sections/proofs/oddoptimal}

\subsubsection*{Case when the optimum is $\pendant{k}+1$}

Going forward, we can assume that the input graph $G$ has even number of degree-2 vertices, since otherwise~\cref{lem_oddoptimal} is applicable. Below, we discuss a structural case when at least $\pendant{k}+1$ edges are required in an optimal solution.

\begin{restatable}{proposition}{bigarc}\label{prop_bigarc}
Let $G$ be an instance of \textsc{OTA-Fix} and let $F$ be a face containing two consecutive degree-2 vertices $v_1$ and $v_2$ (i.e., no other degree-2 vertex lies on the boundary of $F$ between them). If $t(v_1,v_2) \ge \pendant{k}$ and $k$ is even, then at least $\pendant{k}+1$ edges are required to triconnect $G$.
\end{restatable}

\input{Sections/proofs/bigarc}

\begin{algorithm}[t]
\caption{\textsc{TriConnExact}}
\label{triconnexact}

\KwIn{Outerplane 2-connected graph $G$}
\KwOut{List of edges $E'$ such that $G(V, E \cup E')$ is planar and triconnected}

Compute $T(G)$, $k \gets \text{no. of leaves}$, and $d \gets \max_{v \in T(G)} \deg(v)$\;

\If{$T(G)$ is balanced or $k$ is odd}{
    \Return{\textsc{TriConn}($G$)} \tcp*[r]{\Cref{lem_balancedopt} and \Cref{lem_oddoptimal}}
}

Let $v_F$ be the vertex corresponding to face $F$ with maximum-degree $d$ in $G$\;

\If{$F$ has $\le 4$ degree $2$ vertices on it}{
    Brute-force internal (vertex-disjoint) edge addition in $F$ with a budget of at most $2$ non-crossing edges between non-adjacent degree $2$ vertices\;
    \If{there exists a set of edges $M$ such that $G + M$ is balanced}{
        \Return{$M \cup \textsc{OuterPlanarConnect}(G + M)$}\;
    }
    \Return{\textsc{TriConn}($G$)}\tcp*[r]{\Cref{lem:smallfacecase}}
}

Let $v_1, \hdots, v_r$ (with $v_{i \bmod r}$) be the degree $2$ vertices in $F$ in consecutive cyclic order\;

\For{$i \in [r]$}{
    \If{$t(v_i, v_{i+1}) \ge \pendant{k}$}{
        \Return{\textsc{TriConn}($G$)} \tcp*[r]{\Cref{prop_bigarc}}
    }
    \ElseIf{$t(v_i, v_{i+1}) = \pendant{k} - 1$}{
        \Return{$\{v_i v_{i+1}\} \cup \textsc{TriConn}(G + v_i v_{i+1})$}\tcp*[r]{\Cref{observation:balancingedge}}
    }
}

Let $u$ be a degree $2$ vertex of $F$\;

\eIf{there exists a degree $2$ vertex $v$ in $F$ such that $|t(u,v) - t(v,u)| \le 1$}{
    \Return{$\{uv\} \cup \textsc{OuterPlanarConnect}(G + uv)$}\;
}{
    \Return{\textsc{FiveVertices}($G,u$)}\;
}

\end{algorithm}

\subsubsection*{Easy Optimal cases}
In some configurations, it is possible to save case analysis. We shall now focus on the case where it is easy to compute the optimal solution.

\begin{restatable}[Small-face case]{lemma}{smallfacecase}\label{lem:smallfacecase}
Let $F$ be a maximum-degree face containing at most four degree-2 vertices.
Then \textsc{TriConnExact} returns an optimal augmentation for OTA-Fix.
\end{restatable}
\begin{proof}
\Cref{triconnexact} enumerates all possibilities of inserting one or two internal edges
between degree-2 vertices in $F$.

Suppose the optimal solution has size $\pendant{k}$. Since $k$ is even,
any optimal solution connects only degree-2 vertices. By
\Cref{internaledgeoneface}, there exists an optimal solution whose internal
edges all lie inside $F$. Hence, the brute-force step exhausts all possible
internal-edge configurations of an optimal $\pendant{k}$-edge solution,
and the algorithm finds such a solution.

Otherwise, every set of internal edges leaves the graph unbalanced. By
\Cref{degreegoverns}, any triconnectivity augmentation then requires at
least $\pendant{k}+1$ edges. In this case, \textsc{TriConn} produces an optimal
solution of size $\pendant{k}+1$ (\Cref{theorem_almostopt}).

Thus, in both cases the algorithm outputs an optimal solution.
\end{proof}

The remaining easy case where $t(v_i,v_{i+1}) = \pendant{k} - 1$ for two degree-2 vertices $v_i$ and $v_{i+1}$ occurs as a continuation of analysis from \Cref{prop_bigarc}. Here, one internal edge suffices to balance the graph.

\begin{observation}\label{observation:balancingedge}
    If $t(v_i,v_{i+1}) = \pendant{k} - 1$ for two degree-2 vertices $v_i$ and $v_{i+1}$, then $G+v_iv_{i+1}$ is balanced and \tc\,gives an optimal solution. 
\end{observation}

The detailed procedure is given in Algorithm~\ref{triconnexact}.

\begin{algorithm}[t]
\caption{\textsc{FiveVertices}}
\label{fivevertices}

\KwIn{Outerplane 2-connected graph $G$, a degree-2 vertex $u$ in the maximum degree face $F$}
\KwOut{List of edges $E'$ such that $G(V, E \cup E')$ is planar and triconnected}

Compute $T(G)$, $k \gets \text{no. of leaves}$, and $d \gets \max_{v \in T(G)} \deg(v)$\;
Let $v$ be a degree $\ge 3$ vertex in $F$ such that $|t(u,v)-t(v,u)| \le 1$\;
Let $w,x$ be the degree $2$ vertices in $F$ encountered moving anticlockwise and clockwise from $v$ along the boundary of face $F$\;
Let $y,z$ be any two other degree $2$ vertices on $F$\;
Let $u_1,u_2,u_3,u_4,u_5$ be the vertices $u,w,x,y,z$ in consecutive clockwise cyclic order along the face such that $t(u_1,u_2) \ge \max_i t(u_i,u_{i+1})$\;

\If{there exists a face $F'$ in $G$ such that $d(F') = \pendant{k}$}{
    \Return{$\{u_2u_4\} \cup \textsc{OuterPlanarConnect}(G+\{u_2u_4\})$} \tcp*[r]{\Cref{lem:tightfaceshortcut}}
}

$M \gets \emptyset$\;
Let the input in the given cyclic order be $a_1,a_2,a_3,a_4,a_5$ where $t(u_i,u_{i+1}) = a_{i+1}$ and $a_1 = \max_i a_i$\;

\If{$a_1 = 0$}{
    $M \gets \{u_2u_4\}$\;
}
\ElseIf{$a_4+a_3 \le a_1+a_3+a_5$ and $a_2+a_5 \le a_1+a_3+a_4$}{
    $M \gets \{u_1u_5, u_2u_4\}$\;
}
\ElseIf{$a_4+a_3 = a_1+a_2+a_5+1$}{
    $M \gets \{u_2u_4\}$\;
}
\ElseIf{$a_2+a_5 = a_1+a_3+a_4+1$}{
    $M \gets \{u_5u_4, u_1u_3\}$\;
}
\ElseIf{$a_4+a_3 \ge a_1+a_2+a_5+2$}{
    $M \gets \{u_1u_5, u_2u_3\}$\;
}
\ElseIf{$a_2+a_5 \ge a_1+a_3+a_4+2$}{
    $M \gets \{u_1u_2, u_5u_4\}$\;
}

\Return{$M \cup \textsc{OuterPlanarConnect}(G+M)$} \tcp*[r]{\Cref{lem:fivecasesbalance}}

\end{algorithm}

\subsubsection*{The Final Balancing Case}

\begin{figure}
  \centering
  \begin{minipage}[t]{0.33\textwidth}
    \centering
    \includegraphics[width = 1.1\linewidth]{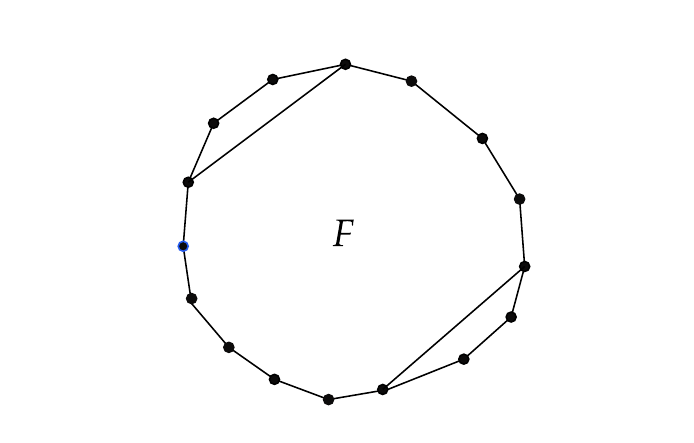}
    \smallskip

    \textbf{(a)}
    \label{fig:a}
  \end{minipage}\hfill
  \begin{minipage}[t]{0.33\textwidth}
    \centering
    \includegraphics[width = 1.2\linewidth]{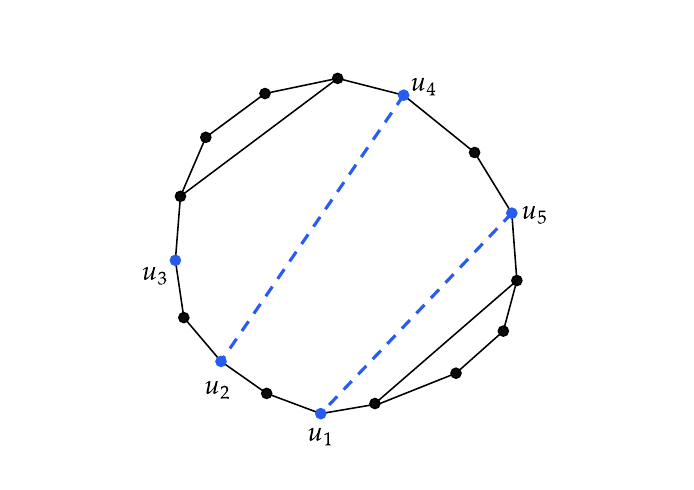}
    \smallskip

    \textbf{(b)}
    \label{fig:b}
  \end{minipage}
    \begin{minipage}[t]{0.33\textwidth}
    \centering
    \includegraphics[width = 1.2\linewidth]{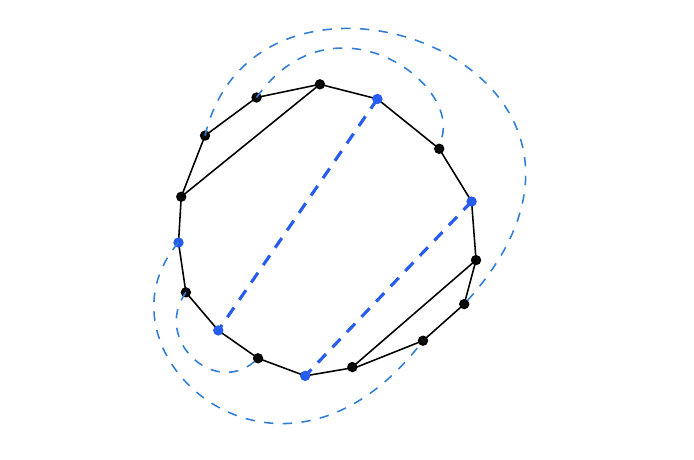}
    \smallskip

    \textbf{(c)}
    \label{fig:c}
  \end{minipage}

  \caption{\textbf{(a)} The input graph with the maximum-degree face $F$.
\textbf{(b)} Five representative degree-2 vertices selected on the boundary of $F$. The thick dotted blue lines show the internal augmentation edges, resulting in a balanced graph.
\textbf{(c)} Augementation of the resultant balanced graph using outer edges.}
  \label{fig:procedure}
\end{figure}

The cases that we have discussed until now give us the following entry condition on the maximum-degree face $F$: $t(v_i,v_{i+1}) \le \lceil k/2\rceil - 2$ for all $i$, there are at least five degree-2 vertices on the boundary of $F$ and for a fixed degree-2 vertex $u$ on the boundary of face $F$ there exists no degree-2 vertex $u'$ satisfying
$|t(u,u')-t(u',u)|\le 1$. Under these conditions, we show that it is possible to pick five degree-2 vertices from the boundary of $F$ such that they represent the face, i.e., every internal augmentation edge is incident on one of these vertices. Once these vertices are chosen, we carefully analyze the sizes of the facial segments (captured by the values $t(v_i,v_{i+1})$) and accordingly add one or two edges inside $F$. We then show that the resulting graph is balanced, after which the resulting graph can be made triconnected using \textsc{OuterPlanarConnect}. $\mathrm{OPT}(G)=\lceil k/2\rceil$ edges. 

The following lemmas establish this in stages. 
We first show that it is always possible to select a suitable set of five vertices (\cref{lem:fivevertices}). 
We then analyze how adding one or two edges inside the face affects the degrees of the resulting faces (\Cref{lem:oneedgebalance,lem:twoedgebalance}). 
Finally, we combine these results to argue that the configurations considered by \textsc{FiveVertices} always lead to a balanced graph (\Cref{lem:fivecasesbalance}). Refer to \Cref{fig:procedure} for an illustrative example.

\begin{restatable}{lemma}{fivevertexlemma}\label{lem:fivevertices}
The five degree-2 vertices
$u_1,\ldots,u_5$ selected by the balancing procedure satisfy
\[
t(u_i,u_{i+1}) \le \pendant{k}-2
\]
for every consecutive pair.
\end{restatable}

\begin{proof}   
Since there is no degree-2 vertex $u'$ on $F$ with
$|t(u,u')-t(u',u)|\le 1$, the balanced split of $F$ from $u$ must be
achieved by a boundary vertex $v$ of degree at least three. We can make the following observation about $v$: since $k$ is even and $t(u,v)+t(v,u) < d(F) \le \pendant{k}$ such that $|t(u,v)-t(v,u)| \le 1$, we can conclude that $t(u,v)\le \pendant{k}-1 \text{ and } t(v,u) \le \pendant{k}-1$. Let $w$ be the last degree-2 vertex before $v$ in clockwise order, and let $x$ be the first degree-2 vertex after $v$ we can say that $t(u,w)\le \pendant{k}-2$ (since $t(u,w) < t(u,v)$) and $t(x,u) \le \pendant{k}-2$ (since $t(x,u) < t(v,u)$). We also know that $t(w,x) \le \pendant{k}-2$ since $w$ and $x$ are consecutive degree-2 vertices. Since $F$ contains at least five degree-2 vertices, pick degree-2 vertices $y$ and $z$ arbitrarily. These 5 vertices are identified as $\{u_i : i \in [5]\}$ in cyclic order. For every consecutive pair among
these five vertices, the corresponding boundary interval is contained in one of the intervals already bounded above. From the said arguments, we obtain $t(u_i,u_{i+1})\le \pendant{k} - 2$.
\end{proof}

Recall that all internal augmentation edges are added inside the maximum-degree face $F$. To ensure this, we first establish that every other face has sufficiently small degree to satisfy the entry condition for the remaining case analysis. The following observation shows that if some interior face violates this condition, then adding the edge $u_2u_4$ inside $F$ immediately balances the graph.

    \begin{observation}\label{lem:tightfaceshortcut}
If some interior face $H$ satisfies $d(H)=\pendant{k}$, then adding the edge $u_2u_4$ inside $F$
produces a balanced graph.
\end{observation}
\begin{proof}
    From \cref{lem:bigfacesmallface}, we gather that $d(F') \le \pendant{k}$. Say the equality obtained for some face $H$, i.e., $d(H) = \pendant{k}$. After the edge $u_2u_4$ is added to the face $F$, the new faces $F_1$ and $F_2$ have $d(F_i) \le \pendant{k}$. Since $d(F') \le \pendant{k}$ for all other faces $F'$, the resultant graph is balanced.
\end{proof}

	Let us recall that for all faces $F$, $d(F) \le k$. Indeed, in any tree, the number of leaves is at least the degree of the maximum degree node. Specifically in $T(G)$, $k \ge d(v_F) \text{ for all }  v_F \in T(G)$. Let us prove specific conditions when adding one or two edges to $F$ gives balance.
    
\begin{restatable}{lemma}{oneedgebalance}[Balancing after one internal edge]\label{lem:oneedgebalance}
Let a face $F$ of degree $a+b+2 \le k$ be split by adding one internal edge into faces
$F_1$ and $F_2$ with $d(F_1)=a+1$ and $d(F_2)=b+1$, where $a\ge b$.
If $a \le b+1$, then
\[
d(F_1)-1 \le \pendant{k-2}
\quad\text{and}\quad
d(F_2)-1 \le \pendant{k-2}.
\]

\end{restatable}
\begin{proof}
    Consider the scenario where a singular edge $e$ is added to the degree-2 vertices to get two new faces $F_1$ and $F_2$ with degrees $d(F_1)=a+1$ and $d(F_2)=b+1$. So, $d(F)=a+b+2 \le k$. WLOG, let $a\ge b$. Then,
	
	\begin{equation}\label{eq1}
		a \le b+1 \iff a \le \pendant{a+b} \le \pendant{k-2}
	\end{equation}
	
	Hence, $d(F_1)-1 \le \pendant{k-2}$ and $d(F_2)-1 \le \pendant{k-2}$. 
\end{proof}

\begin{restatable}{lemma}{twoedgebalance}[Balancing after two internal edges]\label{lem:twoedgebalance}
Suppose two internal edges added inside a face $F$ create faces
$F_1,F_2,F_3$ with degrees $a+1$, $b+2$, and $c+1$, respectively,
where $a+b+c+4 \le k$.
If
\[
a \le \pendant{a+b+c}, \quad
b+1 \le \pendant{a+b+c}, \quad
c \le \pendant{a+b+c},
\]
then all new faces satisfy $d(F_i)-1 \le \pendant{k-4}$.

\end{restatable}
\begin{proof}
    	Consider the scenario where the edges $e_1$ and $e_2$ are added to obtain faces $F_1, F_2, F_3$ with degrees $a+1, b+2, c+1$. For the resultant graph to be balanced, we have to check the degrees of 3 faces.
	
	\begin{equation}\label{2a}
		d(F_1)-1 = a \le \pendant{a+b+c} 
	\end{equation}
	\begin{equation}\label{2b}
		d(F_2)-1 = b+1 \le \pendant{a+b+c} 
	\end{equation}
	\begin{equation}\label{2c}
		d(F_3)-1 = c \le \pendant{a+b+c}  
	\end{equation}
	
	We know that $k \ge a+b+c+4 = d(F)$. 
	These conditions are sufficient to obtain $d(F_i) - 1\le \pendant{k-4}$, that is, these are sufficient conditions to obtain a balanced graph.
 \end{proof}

We have built machinery to understand how edge additions affect the degree of the new faces and the balance of a graph. Let $a_i := t(u_i,u_{i+1})$ and assume $a_1=\max_i a_i$. What follows is a lemma that affirms that the graph is balanced after the various edge additions on the different possible configurations in \textsc{FiveVertices}.

\begin{restatable}{proposition}{fivecases}[\textsc{FiveVertices} cases yield balance]\label{lem:fivecasesbalance}
Let $u_1,\dots,u_5$ and $a_1,\dots,a_5$ be as in \cref{lem:fivevertices},
and assume $a_i \le \pendant{k}-2$ for all $i$.
In each of Case of \textsc{FiveVertices}, the edges inserted satisfy
the hypotheses of \cref{lem:oneedgebalance} or \cref{lem:twoedgebalance}.
Therefore, each case produces a balanced graph.
\end{restatable}
\begin{proof}
We exhaustively consider the various cases as follows:

    \textbf{Case 1:} $a_4+a_3 \le a_1+a_2+a_5$ and $a_2+a_5 \le a_1+a_3+a_4$.

Edges $u_1u_5$ and $u_2u_4$ are inserted, creating three faces.
Using $\sum_{i=1}^5 a_i + 5 = d(F) \le k$ and the case inequalities,
we obtain
\[
d(F_i)-1 \le \pendant{k-4}
\quad \text{for } i=1,2,3.
\]
Thus, the hypotheses of \cref{lem:twoedgebalance} hold.

\textbf{Case 2:} $a_4+a_3 = a_1+a_2+a_5+1$.

The edge $u_2u_4$ is inserted, splitting $F$ into two faces.
The case condition implies $a_4+a_3+1 \le \pendant{k-2}$,
so \cref{lem:oneedgebalance} applies.

\textbf{Case 3:} $a_2+a_5 = a_1+a_3+a_4+1$.

Edges $u_5u_4$ and $u_1u_3$ are inserted, creating three faces $F_1,F_2,F_3$.
From the case assumption and $a_1=\max_i a_i$, we obtain
\[
d(F_1)-1 \le a_2+a_3+1,\quad
d(F_2)-1 \le a_1+a_4+1,\quad
d(F_3)-1 \le a_5.
\]
Using $a_2+a_5 = a_1+a_3+a_4+1$ and $\sum_{i=1}^5 a_i + 1 = d(F) \le k$,
we derive
\[
\max_i (d(F_i)-1) \le a_2+a_5 \le \pendant{k-4}.
\]
Thus the hypotheses of \cref{lem:twoedgebalance} hold.

\textbf{Case 4:} $a_4+a_3 \ge a_1+a_2+a_5+2$.

Edges $u_1u_5$ and $u_2u_3$ are inserted. The resulting faces satisfy
\[
d(F_1)-1 \le a_1,\quad
d(F_2)-1 \le a_2+a_4+a_5+1,\quad
d(F_3)-1 \le a_3.
\]
Since $a_1=\max_i a_i$ and $\sum_{i=1}^5 a_i+1 = d(F) \le k$,
the case inequality implies
\[
a_2+a_4+a_5+1 \le \pendant{k-4}.
\]
Hence all faces satisfy $d(F_i)-1 \le \pendant{k-4}$,
and \cref{lem:twoedgebalance} applies.

\textbf{Case 5:} $a_2+a_5 \ge a_1+a_3+a_4+2$.

Edges $u_1u_2$ and $u_5u_4$ are inserted. The resulting faces satisfy
\[
d(F_1)-1 \le a_2,\quad
d(F_2)-1 \le a_1+a_3+a_4+1,\quad
d(F_3)-1 \le a_5.
\]
Using the case inequality and $\sum_{i=1}^5 a_i+1 = d(F) \le k$,
we obtain
\[
a_1+a_3+a_4+1 \le \pendant{k-4}.
\]
Thus all faces satisfy $d(F_i)-1 \le \pendant{k-4}$,
and \cref{lem:twoedgebalance} applies.

In all five cases the hypotheses of \cref{lem:oneedgebalance}
or \cref{lem:twoedgebalance} are satisfied, and therefore the
resulting graph is balanced.
\end{proof}

\begin{observation}\label{observation:embeddingpreserve}
    All edges inserted by \textsc{FiveVertices} can be drawn inside 
$F$ without crossings; hence the embedding and planarity is preserved.
\end{observation}

 \begin{remark}
	There exists an optimal solution that adds at most 2 internal edges. However, there are instances where at least 2 internal edges are required to be added.
\end{remark}

We can now finally prove the main theorem.

 \triconnmainthm*

     \begin{proof}
Let $G$ be the input graph with $k$ degree-2 vertices.
By \cref{lem_oddoptimal}, we may assume $k$ is even.
Any augmentation requires at least $\pendant{k}$ edges,
and in some configurations at least $\pendant{k}+1$ edges
(\cref{prop_bigarc}).

Let $F$ be a maximum-degree face and $v_1, \hdots, v_r$ (with $v_{i \bmod r}$) be the degree-$2$ vertices in $F$ in consecutive cyclic order.

\medskip
\noindent\textbf{Case 1: $T(G)$ is balanced.}
Then \textsc{OuterPlanarConnect} produces a solution of size
$\pendant{k}$, which is optimal (\cref{lem_balancedopt}).

\medskip
\noindent\textbf{Case 2: $T(G)$ is unbalanced.}

\textbf{Case 2.1: $F$ contains at most four degree-2 vertices.}
By \cref{lem:smallfacecase}, \textsc{TriConnExact}
returns an optimal solution.

\smallskip
\textbf{Case 2.2: $F$ contains at least five degree-2 vertices.}

\emph{Subcase 2.2.1:} There exist consecutive degree-2 vertices
$v_i,v_{i+1}$ with $t(v_i,v_{i+1}) \ge \pendant{k}$.  
Then at least $\pendant{k}+1$ edges are required
(\Cref{prop_bigarc}), and \textsc{TriConn}
computes an optimal solution.

\emph{Subcase 2.2.2:} There exist consecutive degree-2 vertices
$v_i,v_{i+1}$ with $t(v_i,v_{i+1}) = \pendant{k}-1$.  
Then \cref{observation:balancingedge} completes the augmentation
using $\pendant{k}$ edges.

\emph{Subcase 2.2.3:} For all consecutive degree-2 vertices $v_i$ and $v_{i+1}$, we have
$t(v_i,v_{i+1}) \le \pendant{k}-2$.  
By \cref{lem:fivevertices}, the subroutine
\textsc{FiveVertices} selects five suitable degree-2 vertices in $F$.  
\Cref{lem:tightfaceshortcut} ensures that for the remaining case analysis, all the internal faces apart from $F$ have sufficiently small degree. Coupled with the guaranties from
Lemmas~\ref{lem:oneedgebalance} and \ref{lem:twoedgebalance}, 
\cref{lem:fivecasesbalance} shows that the edges inserted by
\textsc{FiveVertices} produce a balanced graph. The remaining
augmentation is then optimally completed by 
\textsc{OuterPlanarConnect}. 

\medskip
In all cases, the number of edges inserted equals the minimum
possible. Embedding preservation follows from
\cref{observation:embeddingpreserve}.

\paragraph*{Running Time and Space.}
All structural preprocessing steps, including the construction of
$T(G)$ and identification of the maximum-degree face, can be performed
in linear time and space.
The algorithm then performs a constant number of structural tests on
face $F$, each requiring at most linear time. In the small-face case,
the brute-force step examines only a constant number of edge insertions.
The checks for large arcs and balanced split pairs are linear-time
operations.

In the remaining case, \textsc{FiveVertices} selects five vertices and
inserts at most two internal edges, which takes linear time. Once the
graph becomes balanced, the remaining augmentation is performed by
\textsc{OuterPlanarConnect}, which runs in
$\mathcal{O}(|V(G)|(1+\alpha(|V(G)|)))$ time (\cref{lem_balancedopt}).
The procedure \textsc{TriConn} is invoked only in configurations where
at least $\pendant{k}+1$ edges are required, and by
\cref{theorem_almostopt} also runs in
$\mathcal{O}(|V(G)|(1+\alpha(|V(G)|)))$ time.
No step allocates more than linear auxiliary data structures, and all
subroutines operate in linear space. Hence, the total space usage is
$\mathcal{O}(|V(G)|)$.
\end{proof}

%% file: Sections/proofs/tritotrideg.tex
\begin{proof}
	Let $G$ be considered to be an Outercircle $C$ with interior chords. $G^* = G + A_{opt}$ is triconnected. If $k$ is even, then all the edges in $A_{opt}$ must connect degree 2 vertices in $G$. If $k$ is odd, there exists at most one edge $e$ connecting a degree 2 vertex $u$ to a degree $>2$ vertex, say, $w$. Since $\deg(w) >3$, $w$ should be incident on an internal chord, say, $wx$ where $x$ is another vertex. We can find a degree 2 vertex $v$ in the Outercircle between the arc $C_{wx}$. Now, the vertices $u,w,v,x$ are in cyclic ordering on $C$. 
	
	We claim that $A_{tri}=A_{opt} - \{uw\} + \{vw\}$ is an optimal solution as well. Assume on the contrary that in $G' = G + A_{tri}$ is not triconnected. So $G' - \{a,b\}$ is disconnected for some vertices $a$ and $b$. Upon removal of $a$ and $b$ from $G'$, the circle $C$ is split into 2 arcs $P_1$ and $P_2$. If $u$ and $w$ are present in the same arc or if they are a part of the separating pair, we realize that since $G^* -  \{a,b\}$ is not disconnected, there exists another edge $e$ from $P_1$ to $P_2$ in $G^*$ which should exist in $G'$ as well. This ensures that $P_1$ and $P_2$ stay connected in $G'$ thereby giving us a contradiction that $\{a,b\}$ is not a separating pair. Hence, assume that $u$ and $w$ are in different arcs.
	
	Hence, let $a$ be in the arc $C_{uw}$\footnote{Referring to the clockwise part of the cycle $C$ restricted between $a$ and $b$}. The other vertex $b$ can be in the arc $C_{ux}, C_{xv}$ or $C_{vw}$. We'll show that for each possibility, there exists a contradiction.
	
	If $b$ is in the arc $C_{ux}$, then $uv$ connects $P_1$ and $P_2$. If $b$ is in the arc $C_{xv}$, then either $xw$ or $uv$ connects $P_1$ and $P_2$. If $b$ is in the arc
	$C_{vw}$ then $xw$ connects $P_1$ and $P_2$. Hence, we've arrived at a contradiction for all the cases exhausted and $G' - \{a,b\}$ remains connected. Since $a,b$ were arbitrary, $G'$ is triconnected. The running time and space complexity come from \Cref{lem_onlyleaves}.
\end{proof}

%% file: Sections/proofs/outerktwo.tex
\begin{proof}
	Let $A$ represent a set of optimal triconnecting edges for $G$, that is, $G+A$ is triconnected and $|A| = \pendant{k}$.
	
	Due to \cref{tritotrideg2} we can safely assume that for each edge in $A$, its endpoints are both of degree $2$. Then, we have for each edge $e$ in $A$, a counterpart on $T(G)$ (the edge between the two pendants in $T(G)$ that correspond to the degree $2$ vertices that are endpoints of $e$ in $G$. Then, we can see that the graph $T(G)+A$ is well defined.
	We claim that $T(G)+A$ must be $2$-vertex connected.
	\begin{claim}
		$T(G)+A$ is 2-vertex connected.
	\end{claim}
	\begin{claimproof}
			If not, let $v_F \in V(T(G)+A)$ be a cut vertex. Consider the bc-tree of $T(G)+A$ that is rooted at $v_F$. Let $a_1,a_2,\ldots,a_m$ be the neighbours of $v_F$ such that they appear in that cyclic order around $v_F$ in the `induced' embedding of $T(G)$.
			
			Consider the bc-tree of $T(G)+A$ to be rooted at $v_F$. Then focusing on the neighbours of $v_F$, $a_1,a_2 \dots , a_m$ (that appear in that cyclic order in the embedding of $T(G)$), it must be the case that the partition of the these vertices based on the connected component they belong to in $T(G)+A-v_F$ is a cyclic interval partitioning. That is, if $a_{i_1},a_{i_2},a_{i_3},a_{i_4}$ appear in that cyclic order then, it can't be the case that $a_{i_1}$ and $a_{i_3}$ belong to one connected component and $a_{i_2}$ and $a_{i_4}$ belong to a different connected component of $T(G)+A-v_F$. This is because, all the edges in $A$ are drawn externally in for the augmentation of $G$, and thus an interleaving will indicate violation of the planarity of $G+A$. 
			
			Now, we pick one connected component of $T(G)+A-v_F$. Let the vertices from $a_1,a_2,\ldots,a_m$ that are picked in this connected component be $a_{i_1}, a_{i_2},\ldots, a_{i_\ell}$ (in the same cyclic order as in $a_1,\ldots,a_m$). Our goal now is to a find a separating pair in $T(G)+A$ to show a contradiction. For some $j\in[\ell]$, let $a_{i_j},a_{i_{j+1}}$ be two vertices that do not consecutively appear around $v_F$. We identify a vertex $s_x$ in $G$, for $a_x$ (for $x\in\{a_{i_j},a_{i_{j+1}}\}$ as follows. If $a_x$ is a leaf vertex in $T(G)$ then $s_x$ is the degree $2$ vertex in $G$ that corresponds to $a_x$. Otherwise, $a_x$ is a face vertex that is adjacent to $v_F$. Exactly two vertices of this face lie on $F$ as well. For $a_{i_j}$ we pick the last vertex and for $a_{i_{j+1}}$ we pick the first vertex (first and last as per the cyclic order of vertices on the outer face in the embedding of $G$). We claim that these two vertices form a separating in $G+A$. And the reasoning is the following. There must exist an edge in $A$ that connects the two distinct subtrees of $v_F$ that contain $a_{i_j}$ and $a_{i_{j+1}}$. This edge when drawn on the outerplane drawing blocks any pendant vertex that lies between $s_x$s (the two vertices we delete) to be connected beyond, since we do not have any edge to add internally, and external connection will violate planarity.

			Hence $T(G) + A$ is $2$-connected. The claim is proven.
		\end{claimproof}

		Assume on the contrary that $G$ is not balanced. So, $\max_v \deg(v)-1 > \pendant{k}$ edges are required to augment $T(G)$ to biconnectivity. However, $T(G)+A$ is biconnected (|A| = $\pendant{k})$. This is a contradiction, and therefore, our assumption was wrong, and $G$ is balanced.
	\end{proof}

%% file: Sections/proofs/internaledgeoneface.tex
\begin{proof}
	Let $G$ be an input to OTA-Fix and let $k$ be the number of degree-2 vertices.

If $bc(T(G))$ is balanced, then by \cref{lem_balancedopt},
\textsc{OuterPlanarConnect} yields an optimal solution using $\pendant{k}$
edges and no internal edges. Hence the statement holds trivially.
Assume therefore that $bc(T(G))$ is unbalanced, and let $F$ be a
maximum-degree face.

First consider the case where the optimal solution has size $\pendant{k}+1$.
By \cref{theorem_almostopt}, \textsc{TriConn} computes an optimal
augmentation of this size, inserting exactly one internal edge.
This edge is placed inside $F$ by construction, and thus the claim holds.
	
	Let's assume now that the optimal solution uses $\pendant{k}$ edges, while $bc(T(G))$ is unbalanced: 
	Let $F$ represent the maximum degree face. We claim that there must be a degree $2$ vertex in $F$. This is because, if this is not the case, then no internal edge can be added to $F$ (any internal edge insertion would be suboptimal, because while it will not bring down the number of degree $2$ vertices in $G$, but at the same time it will take up one edge from the budget of the optimal solution). So, we have established that $F$ can't have any internal edges in the optimal solution. Now, consider all the internal edges (say, they are $\ell$ in number) that the optimal solution has to add to be already added. Let this intermediate graph be $G'$. $G'$ must have either $k-2\ell$ or $k-2\ell+1$ degree $2$ vertices remaining. By assumption, since the optimal solution size is $\pendant{k}$, exactly $\pendant{k}-\ell$ are required to triconnect $G'$ by insertion of only outer edges. Since $G'$ is outerplanar, from \cref{outerk2}, we have that $G'$ must be a balanced graph. But, since $F$ is the maximum degree face in $G$, and remains so in $G'$ because all the other face degrees and number of degree $2$ vertices has only gone down as compared to $G$, while degree of face $F$ hasn't changed. This means $G'$ is still unbalanced, and we have a contradiction. So, from now on we can safely assume that $F$ has at least one degree $2$ vertex and moreover, any optimal solution (of size $\pendant{k}$) must have an edge added internally in $F$.    
	
	Now, if $k$ is odd, then performing \cref{triconn} gives an optimal solution by adding an interior edge in $F$ incident on a degree $2$ vertex. Indeed, the remaining $k-1$ (or $k-2$) vertices are connected by $\pendant{k-1}$ (or $\pendant{k-2}$) edges so that in total $\pendant{k}$ edges are used. Hence, we are done in this case as well.
	
	Consider the case when $k$ is even. Let $I$ be the set of edges in the optimal solution that are added inside $F$. $I$ is not empty from  the discussion above. 
	Since $k$ is even, each degree $2$ vertex in $F$ is incident on exactly a single edge from $I$.
	We can now assume that $bc(T(G + I))$ is not balanced. Because if it is not the case, then we can optimally connect the remaining degree $2$ vertices in $G+I$ using \textsc{OuterPlanarConnect} and we will obtain an optimal solution with all internal edges only being in $F$ and would be done.
	
	Let $F'$ be the maximum degree face in $G+I$. It is unique since $bc(T(G + I))$ is unbalanced.
	
	\begin{enumerate}[wide, labelindent=0pt,label=Case~\arabic*.]
		\item {$F'$ is a new face}\label{caselemma1}:
		
		By assumption, insertion of further internal edges does not split this maximum degree face $F'$. Hence, after all the internal edges have been inserted, the resulting graph still remains unbalanced. But, it can be triconnected optimally with only outer edges. This contradicts \Cref{degreegoverns}.

		\item {$F'$ is an old face}\label{caselemma2}:
		
		$F'\neq F$ and is an internal face in $G$. Let $I=\{e_1,e_2,\ldots,e_k\}$. For every $j\in[k]$ let $I_j=\{e_1,\ldots,e_j\}$. Consider the following sequence of graphs, $G+I_1, G+I_2,\ldots,G+I_k=G+I$. By assumption, $F'$ is the maximum degree face in $G+I$. Let $\ell$ be the smallest index in $[k]$ such that $F'$ is the maximum degree face in $G+I_\ell$. Then we claim that $F'$ is the maximum degree face in all of $G+I_{\ell+1},G+I_{\ell+2},\ldots,G+I_{k-1}$ as well. Indeed, the degree of $F'$ doesn't change through these edge additions and no new face gets a high degree.
		
		Let $G_1=G+I_{\ell-1}$ and $G_2=G+I_{\ell}$. In $G_1$, the maximum degree face must be a face that lies inside the boundary of erstwhile $F$. In $G_2$ by assumption, $F'$ becomes the maximum degree face. Let $M$ be the face in $G+I_{\ell-1}$ to which the edge $e_\ell$ is added to get $G+I_{\ell}$. If $G_1$ is already balanced, we are done since we add the $l-1$ edges to $F$ and employ \Cref{lem_balancedopt} to use $\pendant{k}$ edges in total. Otherwise, we can employ \cref{lem_newmaxface} to the graph $G_1$ upon adding the edge $e$ since an old face gets the maximum degree after adding $e$. This gives us that $G_2$ is balanced. 
		
		In $G_2$, there are $k-2\ell$ degree $2$ vertices remaining.
		Hence, use the $\ell$ edges from the face $F$ and use \textsc{OuterPlanarConnect} on $G+I_\ell$ to optimally augment the graph. We use $\ell + \pendant{k} - \ell = \pendant{k}$ edges exactly.
		
	\end{enumerate}
	
\end{proof}

%% file: Sections/proofs/oddoptimal.tex
\begin{proof}
	Say the input graph $G$ is balanced, then from \cref{lem_balancedopt}, we know that the algorithm returns the optimal solution. If not, let $F$ be the massive face. 
	Say that face $F$ has a degree 2 vertex. Then the algorithm adds an interior edge by connecting a degree 2 vertex and another vertex to split the face $F$. Once this edge is added, we know that the graph gets balanced. Hence, in total, 1 + $\pendant{k-1}$ or 1 + $\pendant{k-2}$ are used to triconnect the graph depending on whether the edge $e$ connects two degree 2 vertices or just one. Since $k$ is odd, both the quantities are equal. Hence, an optimum of $\pendant{k}$ edges are used in total.
	
	Now consider the case when $F$ doesn't have a degree 2 vertex. We know from \cref{internaledgeoneface} that internal edges (if any) are added to $F$. If an internal edge was indeed added, we know that it doesn't connect any degree 2 vertex and $k$ degree 2 vertices have to be connected using outer edges requiring at least $\pendant{k}+1$ edges in total. If no internal edge was added, the graph is still unbalanced and from \Cref{degreegoverns}, at least $\pendant{k}$+1 edges are required by the optimal solution. Hence, in totality, \textsc{TreeConn} provides the optimal solution. 
	
\end{proof}

%% file: Sections/proofs/bigarc.tex
\begin{proof}
	Assume that the optimal solution has size $\pendant{k}$. From \Cref{degreegoverns}, since $G$ is unbalanced, the solution must include internal edges.

    \begin{claim}
        The graph $G$ remains unbalanced after adding the internal edges.
    \end{claim}
    \begin{claimproof}
        Since $k$ is even and the optimal solution uses $\pendant{k}$ edges, these edges can only connect the $k$ degree-2 vertices. After inserting the internal edges, all vertices between $v_1$ and $v_2$ along the boundary of $F$ lie on a common face $F'$. By definition of $t(v_1,v_2)$, we have
\[
d(F') \ge t(v_1,v_2) + 1 \ge \pendant{k} + 1.
\]
After the internal edge insertions, at most $k-2$ degree-2 vertices remain. Hence,
\[
d(F') \ge \pendant{k} + 1 \ge \pendant{k'} + 2,
\]
implying that $F'$ is still a massive face and the graph remains unbalanced.
    \end{claimproof}

    By \cref{internaledgeoneface}, we may assume that all internal edges of the optimal solution lie in $F$. Suppose $l$ such edges are added. Then $k-2l$ degree-2 vertices remain. Since the graph is still unbalanced, \Cref{degreegoverns} implies that at least $\pendant{k} - l + 1$ outer edges are required to achieve triconnectivity. Therefore, the total number of edges used is at least
\[
l + \left(\pendant{k} - l + 1\right) = \pendant{k} + 1,
\]
contradicting the assumption that the optimal solution has size $\pendant{k}$.

\end{proof}

%% file: Sections/6.conclusion.tex
\section{Conclusion}

We studied both edge- and vertex-connectivity augmentation under fixed planar embeddings.
For 2-edge-connectivity augmentation of plane graphs (\textsc{PECA-Fix}), we re-introduced and
used the $bt$-tree structure and a labelling scheme that identifies safe edges, yielding a
near-linear-time optimal algorithm. For 3-vertex-connectivity augmentation of biconnected outerplane graphs
(\textsc{OTA-Fix}), we used the dual tree of inner faces and related biconnecting this tree
to triconnecting the original graph. Furthermore, we developed a face-splitting technique to obtain a (+1)-approximate solution and got rid of the (+1) factor with extensive case analysis and new combinatorial insights and thus obtained a near-linear-time solution.

Several questions remain open. Is it possible to solve \textsc{PECA-Fix} with a better running time? Can a 2-edge-connected outerplane graph be augmented to 3-edge-connectivity in near-linear or better time? Our techniques do not extend directly, but may offer useful insights. It would be interesting to extend the problem to  a weighted setting where the solution edges incur a cost.

%% file: cas-refs.bib
@article{DBLP:journals/siamdm/Vegh11,
    author       = {L{\'{a}}szl{\'{o}} A. V{\'{e}}gh},
    title        = {Augmenting Undirected Node-Connectivity by One},
    journal      = {{SIAM} J. Discret. Math.},
    volume       = {25},
    number       = {2},
    pages        = {695--718},
    year         = {2011},
    url          = {https://doi.org/10.1137/100787507},
    doi          = {10.1137/100787507},
    bibsource    = {dblp computer science bibliography, https://dblp.org}
}

@article{watanabe1993minimum,
title = {A minimum 3-connectivity augmentation of a graph},
journal = {Journal of Computer and System Sciences},
volume = {46},
number = {1},
pages = {91-128},
year = {1993},
issn = {0022-0000},
doi = {https://doi.org/10.1016/0022-0000(93)90050-7},
url = {https://www.sciencedirect.com/science/article/pii/0022000093900507},
author = {Toshimasa Watanabe and Akira Nakamura}
}

@book{diestel2017,
author = {Diestel, Reinhard},
title = {Graph Theory},
year = {2017},
isbn = {3662536218},
publisher = {Springer Publishing Company, Incorporated},
edition = {5th}
}

@article{toth2012connectivity,
title = {Connectivity augmentation in planar straight line graphs},
journal = {European Journal of Combinatorics},
volume = {33},
number = {3},
pages = {408-425},
year = {2012},
note = {Topological and Geometric Graph Theory},
issn = {0195-6698},
doi = {https://doi.org/10.1016/j.ejc.2011.09.002},
url = {https://www.sciencedirect.com/science/article/pii/S0195669811001508},
author = {Csaba D. Tóth}
}

@InProceedings{kant1991planar,
author="Kant, Goos
and Bodlaender, Hans L.",
editor="Dehne, Frank
and Sack, J{\"o}rg-R{\"u}diger
and Santoro, Nicola",
title="Planar graph augmentation problems",
booktitle="Algorithms and Data Structures",
year="1991",
publisher="Springer Berlin Heidelberg",
address="Berlin, Heidelberg",
pages="286--298",
doi = {10.1007/BFb0028270},
isbn="978-3-540-47566-8"
}

@InProceedings{gutwenger2009planar,
author="Gutwenger, Carsten
and Mutzel, Petra
and Zey, Bernd",
editor="Fiala, Ji{\v{r}}{\'i}
and Kratochv{\'i}l, Jan
and Miller, Mirka",
title="Planar Biconnectivity Augmentation with Fixed Embedding",
booktitle="Combinatorial Algorithms",
year="2009",
publisher="Springer Berlin Heidelberg",
address="Berlin, Heidelberg",
pages="289--300",
doi = {10.1007/978-3-642-10217-2_29},
isbn="978-3-642-10217-2"
}

@inproceedings{hartmann2012cubic,
  author    = {Tanja Hartmann and
               Jonathan Rollin and
               Ignaz Rutter},
  editor    = {Kun-Mao Chao and
               Hsu-Chun Yen and
               In-Bok Yeom},
  title     = {Cubic Augmentation of Planar Graphs},
  booktitle = {Algorithms and Computation - 23rd International Symposium, ISAAC 2012, Taipei, Taiwan, December 19-21, 2012. Proceedings},
  series    = {Lecture Notes in Computer Science},
  volume    = {7676},
  pages     = {266--276},
  publisher = {Springer},
  year      = {2012},
  doi       = {10.1007/978-3-642-35261-4_43}
}

@article{eswaran1976augmentation,
author = {Eswaran, Kapali P. and Tarjan, R. Endre},
title = {Augmentation Problems},
journal = {SIAM Journal on Computing},
volume = {5},
number = {4},
pages = {653-665},
year = {1976},
doi = {10.1137/0205044},

URL = { 
    
        https://doi.org/10.1137/0205044
    
    

},
eprint = { 
    
        https://doi.org/10.1137/0205044
    
    

}
}

@inproceedings{fialko1998new, author = {Fialko, Sergej and Mutzel, Petra}, title = {A new approximation algorithm for the planar augmentation problem}, year = {1998}, isbn = {0898714109}, publisher = {Society for Industrial and Applied Mathematics}, address = {USA}, booktitle = {Proceedings of the Ninth Annual ACM-SIAM Symposium on Discrete Algorithms}, pages = {260–269}, numpages = {10}, location = {San Francisco, California, USA}, series = {SODA '98}, url = {http://dl.acm.org/citation.cfm?id=314613.314714} }

@INPROCEEDINGS{hsu1991linear,
  author={Hsu, T.-S. and Ramachandran, V.},
  booktitle={[1991] Proceedings 32nd Annual Symposium of Foundations of Computer Science}, 
  title={A linear time algorithm for triconnectivity augmentation}, 
  year={1991},
  volume={},
  number={},
  pages={548-559},
  doi={10.1109/SFCS.1991.185418}}

@article{jackson2005independence,
title = {Independence free graphs and vertex connectivity augmentation},
journal = {Journal of Combinatorial Theory, Series B},
volume = {94},
number = {1},
pages = {31-77},
year = {2005},
issn = {0095-8956},
doi = {https://doi.org/10.1016/j.jctb.2004.01.004},
url = {https://www.sciencedirect.com/science/article/pii/S0095895604001224},
author = {Bill Jackson and Tibor Jordán}
}

@article{ABELLANAS2008220,
title = {Augmenting the connectivity of geometric graphs},
journal = {Computational Geometry},
volume = {40},
number = {3},
pages = {220-230},
year = {2008},
issn = {0925-7721},
doi = {https://doi.org/10.1016/j.comgeo.2007.09.001},
url = {https://www.sciencedirect.com/science/article/pii/S0925772107000934},
author = {M. Abellanas and A. García and F. Hurtado and J. Tejel and J. Urrutia}
}

@InProceedings{gutwengerhardness,
author="Gutwenger, Carsten
and Mutzel, Petra
and Zey, Bernd",
editor="Ngo, Hung Q.",
title="On the Hardness and Approximability of Planar Biconnectivity Augmentation",
booktitle="Computing and Combinatorics",
year="2009",
publisher="Springer Berlin Heidelberg",
address="Berlin, Heidelberg",
pages="249--257",
doi          = {10.1007/978-3-642-02882-3\_25},
isbn="978-3-642-02882-3"
}

@article{watanabe1987edge,
title = {Edge-connectivity augmentation problems},
journal = {Journal of Computer and System Sciences},
volume = {35},
number = {1},
pages = {96-144},
year = {1987},
issn = {0022-0000},
doi = {https://doi.org/10.1016/0022-0000(87)90038-9},
url = {https://www.sciencedirect.com/science/article/pii/0022000087900389},
author = {Toshimasa Watanabe and Akira Nakamura}
}

@article{wolff2012augmenting,
title = {Augmenting the Connectivity of Planar and Geometric Graphs},
journal = {Electronic Notes in Discrete Mathematics},
volume = {31},
pages = {53-56},
year = {2008},
note = {The International Conference on Topological and Geometric Graph Theory},
issn = {1571-0653},
doi = {https://doi.org/10.1016/j.endm.2008.06.009},
url = {https://www.sciencedirect.com/science/article/pii/S1571065308000656},
author = {Ignaz Rutter and Alexander Wolff}
}

@InProceedings{altri,
author="Al-Jubeh, Marwan
and Ishaque, Mashhood
and R{\'e}dei, Krist{\'o}f
and Souvaine, Diane L.
and T{\'o}th, Csaba D.",
editor="Dong, Yingfei
and Du, Ding-Zhu
and Ibarra, Oscar",
title="Tri-Edge-Connectivity Augmentation for Planar Straight Line Graphs",
booktitle="Algorithms and Computation",
year="2009",
publisher="Springer Berlin Heidelberg",
address="Berlin, Heidelberg",
pages="902--912",
doi = {10.1007/978-3-642-10631-6_91},
isbn="978-3-642-10631-6"
}

@article{westbrook1992maintaining,
	title = {Maintaining bridge-connected and biconnected components on-line},
	volume = {7},
	issn = {1432-0541},
	url = {https://doi.org/10.1007/BF01758773},
	doi = {10.1007/BF01758773},
	number = {1},
	journal = {Algorithmica},
	author = {Westbrook, Jeffery and Tarjan, Robert E.},
	month = jun,
	year = {1992},
	pages = {433--464},
}

@article{kant1996augmenting,
author = {Kant, Goos},
title = {Augmenting Outerplanar Graphs},
year = {1996},
issue_date = {July 1996},
publisher = {Academic Press, Inc.},
address = {USA},
volume = {21},
number = {1},
issn = {0196-6774},
url = {https://doi.org/10.1006/jagm.1996.0034},
doi = {10.1006/jagm.1996.0034},
journal = {J. Algorithms},
month = jul,
pages = {1–25},
numpages = {25}
}

@article{hsu1992four,
author = {Hsu, Tsan-sheng},
title = {On Four-Connecting a Triconnected Graph},
year = {2000},
issue_date = {May 2000},
publisher = {Academic Press, Inc.},
address = {USA},
volume = {35},
number = {2},
issn = {0196-6774},
url = {https://doi.org/10.1006/jagm.2000.1077},
doi = {10.1006/jagm.2000.1077},
journal = {J. Algorithms},
month = may,
pages = {202–234},
numpages = {33}
}

@article{garcia2010augmenting,
	title = {Augmenting the {Connectivity} of {Outerplanar} {Graphs}},
	volume = {56},
	issn = {1432-0541},
	url = {https://doi.org/10.1007/s00453-008-9167-1},
	doi = {10.1007/s00453-008-9167-1},
	number = {2},
	journal = {Algorithmica},
	author = {García, A. and Hurtado, F. and Noy, M. and Tejel, J.},
	month = feb,
	year = {2010},
	pages = {160--179},
}

@InProceedings{nagamochi1998edge,
author="Nagamochi, Hiroshi
and Eades, Peter",
editor="Bixby, Robert E.
and Boyd, E. Andrew
and R{\'i}os-Mercado, Roger Z.",
title="Edge-Splitting and Edge-Connectivity Augmentation in Planar Graphs",
booktitle="Integer Programming and Combinatorial Optimization",
year="1998",
publisher="Springer Berlin Heidelberg",
address="Berlin, Heidelberg",
pages="96--111",
doi = {10.1007/3-540-69346-7_8},
isbn="978-3-540-69346-8"
}

@InProceedings{akitaya2025price,
  author =	{A. Akitaya, Hugo and Dallant, Justin and Demaine, Erik D. and Kaufmann, Michael and Kleist, Linda and Stock, Frederick and T\'{o}th, Csaba D. and Ueckerdt, Torsten},
  title =	{{The Price of Connectivity Augmentation on Planar Graphs}},
  booktitle =	{33rd International Symposium on Graph Drawing and Network Visualization (GD 2025)},
  pages =	{23:1--23:24},
  series =	{Leibniz International Proceedings in Informatics (LIPIcs)},
  ISBN =	{978-3-95977-403-1},
  ISSN =	{1868-8969},
  year =	{2025},
  volume =	{357},
  editor =	{Dujmovi\'{c}, Vida and Montecchiani, Fabrizio},
  publisher =	{Schloss Dagstuhl -- Leibniz-Zentrum f{\"u}r Informatik},
  address =	{Dagstuhl, Germany},
  URL =		{https://drops.dagstuhl.de/entities/document/10.4230/LIPIcs.GD.2025.23},
  URN =		{urn:nbn:de:0030-drops-250095},
  doi =		{10.4230/LIPIcs.GD.2025.23}
}

@article{hartmann2015regular,
  title={Regular augmentation of planar graphs},
  author={Hartmann, Tanja and Rollin, Jonathan and Rutter, Ignaz},
  journal={Algorithmica},
  volume={73},
  number={2},
  pages={306--370},
  year={2015},
  publisher={Springer},
  doi = {10.1007/s00453-014-9922-4}
}
